\documentclass[twocolumn]{autart}

\usepackage{cite}
\usepackage{amsmath,amssymb}
\usepackage[table]{xcolor}
\definecolor{proposedgray}{gray}{0.92}
\usepackage{graphicx}
\usepackage{booktabs}
\usepackage{array}
\usepackage{algorithm}
\usepackage{algpseudocode}
\usepackage{flushend}
\usepackage{hyperref}
\hypersetup{hidelinks}
\allowdisplaybreaks
\newcolumntype{L}[1]{>{\raggedright\arraybackslash}p{#1}}

\begin{document}

\begin{frontmatter}

\title{Probe Sets for Nonlinear Kalman Filtering: Multi-Center Gains and Covariance Recalibration\thanksref{footnoteinfo}}

\thanks[footnoteinfo]{This paper was not presented at any IFAC meeting.
Corresponding author Shida Jiang.}

\author[berkeley]{Shida Jiang}\ead{shida\_jiang@berkeley.edu},
\author[chalmers]{Shengyu Tao}\ead{shengyu.tao@chalmers.se},
\author[berkeley]{Scott Moura}\ead{smoura@berkeley.edu}

\address[berkeley]{Department of Civil and Environmental Engineering,
University of California, Berkeley, CA, USA}
\address[chalmers]{Department of Electrical Engineering, Chalmers University of Technology, Gothenburg, Sweden}

\begin{abstract}
Kalman filters (KFs) choose the gain as the product of the state--measurement
cross-covariance and the inverse of the innovation covariance. In nonlinear
systems, the predictive distribution generally changes shape, and these
covariances require approximation. Conventional nonlinear KFs
approximate them around a single center, the predicted state. They can produce
inaccurate gains and overconfident covariance estimates when one local
approximation does not represent the measurement geometry across the
uncertainty region. To address this issue, we introduce probe sets, small covariance-scaled
collections of states at which the filter repeats its local approximation.
We then select the gain that minimizes the average covariance reported by these local approximations.
The same construction is used in our previously proposed covariance
recalibration step to improve covariance consistency. The construction applies to different KF variants without requiring a specific
approximation rule. For quadratic measurements, we derive conditions under
which probing yields lower normalized true error variance than the
corresponding single-center gain when prediction uncertainty is sufficiently
understated. We apply the framework to four KF variants and compare the resulting
filters with multiple baselines across 300 randomized setups for each of two systems. The results show that probing
preserves typical accuracy while substantially reducing covariance
inconsistency and the rare large-error trajectories that dominate the
root mean square of per-run errors. The code is available at \href{https://github.com/Shida-Jiang/Probe_KF}{https://github.com/Shida-Jiang/Probe\_KF}.

\end{abstract}

\end{frontmatter}

\section{Introduction}\label{sec:introduction}

Kalman filters (KFs) select the gain as a function of two moments: the
state--measurement cross-covariance and the innovation covariance. For linear
systems with correctly specified moments, both can be computed exactly, and
the KF minimizes the mean-square state estimation error among all linear estimators \cite{book3}.

Many attempts have been made to extend the KF to handle nonlinear
measurements. For example, the extended KF (EKF) linearizes the measurement
function at the predicted state. In nonlinear systems, however, the
predicted state distribution generally changes shape across time steps, and the
required moments must be approximated. In nonlinear KFs, these
approximate moments determine both the Kalman gain and the reported error
covariance. The gain is optimal for the approximated model \cite{book3},
which need not represent the true predictive distribution. Approximation
error can therefore impair the correction while overstating its accuracy.
Indeed, center-dependent linearization is a documented source of EKF
inconsistency in nonlinear localization \cite{slam2,slam3,Huang2010}.
To capture more of the measurement curvature, other KF variants use more
complicated moment approximation methods. The second-order EKF (EKF2) takes a
second-order approximation of the measurement function and computes its
moments under a Gaussian assumption \cite{book1,book2}. Meanwhile, the
unscented and cubature KFs (UKF and CKF) use weighted point sets to
approximate the moments under an assumed Gaussian distribution
\cite{UKF1,UKF2,CKF}. However, these variants still rely on a single moment approximation centered at the predicted mean. When the measurement is strongly nonlinear across the uncertainty region, the resulting moments can be inaccurate, degrading the state update and causing the reported covariance to understate the actual error.

Our recent work addressed understated covariance estimates through
\emph{covariance recalibration} \cite{parent}. After the state update, it
reevaluates the moments at the updated state using the predicted covariance.
If the resulting report indicates increased uncertainty, a \emph{back-out}
rule discards the update and returns to the predicted state and covariance.
This step revises the covariance report after a correction but retains the
gain selected at the predicted state, which may still be inaccurate.

This paper introduces \emph{probe sets} to improve gain selection. A probe
set contains nearby states whose spread follows the covariance carried by
the filter. The filter evaluates its usual moment rule at each member of the probe set.
It then selects the gain that minimizes the average reported covariance
across the local models. This collective assessment is also adopted in
recalibration. The probe does not rely on a specific moment approximation
rule, so it can be combined with different KF variants.

The proposed probe geometry is related to reduced sigma-point constructions
\cite{Julier03}, but the purposes are different. Divided-difference and spherical-simplex methods use
displaced points to build one derivative or quadrature approximation
\cite{norgaard2000new,Julier03}. The objective is still to approximate the
moments at the center. In the proposed framework, each probe member serves
as the center of a complete local moment approximation. Higher-order Taylor,
sigma-point, and cubature rules use additional computation to refine the
moments at one center \cite{QKF,SQKF,CKF5}, while probing uses it to assess
variation in the local approximations across the uncertainty region. As we show later in the paper, this structural difference yields a concrete advantage when covariance estimates are overconfident. Such overconfidence arises routinely in practice, from strong nonlinearity or systematic modeling error \cite{case1,case2}.

The main contributions are as follows.
\begin{itemize}
\item We challenge the conventional reliance on single-center moment
approximation at the predicted mean and introduce covariance-scaled probe
sets to improve the estimation accuracy and covariance consistency of
nonlinear Kalman filters.
\item For quadratic measurements, we show under the stated assumptions that
probing provides a conservative update-probe covariance report and cannot increase
the true error covariance when the Gaussian prediction is exact. Single-center
updates do not share this guarantee in general. We also establish conditions
under which probing yields lower normalized true error variance than the
center gain of the same moment rule when prediction uncertainty is understated.
\item The framework does not require a specific moment approximation rule
and can be combined with different KF variants. Simulations across two
systems, four KF variants, and 300 randomized setups per system show that
it preserves typical accuracy while improving covariance consistency and
reducing large-error tails. In the low-noise battery cases, the reductions
in the root mean square of per-run errors and overconfidence reach several orders
of magnitude.
\end{itemize}

The rest of the paper is organized as follows. Section~\ref{sec:framework} develops the framework and its theoretical properties,
Section~\ref{sec:results} reports the experimental validation, and
Section~\ref{sec:conclusion} concludes. The appendices give the details of the implemented
KF variants and simulation models.

\section{Probe-set framework}\label{sec:framework}
\subsection{Problem formulation and the single-center limitation}\label{sec:formulation}

Consider the discrete-time system
\begin{equation}\label{eq:system}
 x_k=f_k(x_{k-1},u_{k-1})+w_{k-1},\qquad
 z_k=h_k(x_k)+v_k,
\end{equation}
where $x_k\in\mathbb R^n$ is the state, $u_{k-1}$ is the input, and
$z_k\in\mathbb R^p$ is the measurement. The process and measurement noises are
zero mean and mutually independent, with covariances $Q_{k-1}\succeq0$ and
$R_k\succeq0$. We write $A\succeq B$ when $A-B$ is positive semidefinite, and
$A\succ B$ when it is positive definite. Measurement noise is independent of the predictive state.
We analyze one measurement update conditionally on the past
measurements and suppress the time index. All moments below are assumed
finite. Let $\pi$ be the predictive distribution of
$x$. The filter carries a predicted state $\hat x^-$, a predicted covariance
$P\succ0$, and a predicted measurement $\hat z$.
These quantities are computed before the current measurement.
The carried covariance $P$ need not equal the true predictive covariance.

For a gain $K\in\mathbb R^{n\times p}$ selected without using the current
measurement, the candidate update is
\begin{equation}\label{eq:update}
 \hat x^+=\hat x^-+K\tilde z,\qquad \tilde z:=z-\hat z.
\end{equation}
Define the exact conditional moments
\begin{align}
 P^{\rm t}&:=\operatorname{Var}_\pi(x),&
 C^{\rm t}&:=\operatorname{Cov}_\pi(x,h(x)),\label{eq:true_pair_a}\\
 S^{\rm t}&:=\operatorname{Var}_\pi(h(x))+R.
 \label{eq:true_pair_b}
\end{align}
The central covariance of the candidate estimation error, averaged over the
current state and measurement conditional on past measurements, is
\begin{equation}\label{eq:true_covariance}
 P_+^{\rm t}(K)=P^{\rm t}-KC^{{\rm t}\top}-C^{\rm t}K^\top
 +KS^{\rm t}K^\top.
\end{equation}
This central covariance equals the conditional mean-squared-error matrix
when the update error has zero conditional mean.

A nonlinear KF cannot generally recover $(C^{\rm t},S^{\rm t})$ from the
predicted mean and covariance alone. Instead, it applies a local \emph{moment map}. At an evaluation
center $\xi$ and with spread $P$, the filter returns
\begin{equation}\label{eq:moment_map}
 \Theta_P(\xi):=\big(C_P(\xi),S_P(\xi)\big),\qquad S_P(\xi)\succeq0.
\end{equation}
Here $C_P(\xi)$ approximates the state--measurement cross-covariance, and
$S_P(\xi)$ is the corresponding innovation covariance, including $R$.
Covariance matrices used in inverses are assumed positive definite.
The EKF, EKF2, UKF, and CKF supply different such maps. Appendix
\ref{app:filters} gives their formulas.
A valid moment map must satisfy
\begin{equation}\label{eq:admissible}
 \begin{bmatrix}P&C_P(\xi)\\C_P(\xi)^\top&S_P(\xi)-R\end{bmatrix}\succeq0.
\end{equation}
In light of \eqref{eq:true_covariance}, for a generic moment pair $\theta=(C,S)$, define the generalized Joseph report
\cite{josephUKF,parent}
\begin{equation}\label{eq:generalized_joseph}
 P_{\rm rc}(\theta;K):=P-KC^\top-CK^\top+KSK^\top.
\end{equation}
Condition \eqref{eq:admissible} gives
\begin{equation}\label{eq:joseph_psd}
\begin{aligned}
 P_{\rm rc}(\theta;K)={}&
 \begin{bmatrix}I&-K\end{bmatrix}
 \begin{bmatrix}P&C\\C^\top&S-R\end{bmatrix}
 \begin{bmatrix}I\\-K^\top\end{bmatrix}\\
 &+KRK^\top\succeq0.
\end{aligned}
\end{equation}
The conventional center pair and gain are
\begin{equation}\label{eq:center_pair_gain}
 \theta_0:=\Theta_P(\hat x^-)=(C_0,S_0),\qquad
 K_0:=C_0S_0^{-1}.
\end{equation}
The next identity isolates the covariance penalty caused by an inaccurate
gain and expresses it in terms of moment error \cite{theoretical}.

\begin{prop}[covariance cost of gain error]\label{prop:oracle_regret}
Assume $S^{\rm t}\succ0$ and let
$K^{\rm t}:=C^{\rm t}(S^{\rm t})^{-1}$. For every gain $K$,
\begin{equation}\label{eq:oracle_completion}
 P_+^{\rm t}(K)-P_+^{\rm t}(K^{\rm t})
 =(K-K^{\rm t})S^{\rm t}(K-K^{\rm t})^\top\succeq0.
\end{equation}
Suppose an approximate pair $(\widehat C,\widehat S)$ with
$\widehat S\succ0$ produces
$\widehat K=\widehat C\widehat S^{-1}$. Set
$E_C:=\widehat C-C^{\rm t}$ and $E_S:=\widehat S-S^{\rm t}$. Then
\begin{equation}\label{eq:moment_regret_identity}
 P_+^{\rm t}(\widehat K)-P_+^{\rm t}(K^{\rm t})
 =\Gamma_{\widehat K}(S^{\rm t})^{-1}\Gamma_{\widehat K}^\top,
\end{equation}
where
\begin{equation}\label{eq:effective_moment_error}
 \Gamma_{\widehat K}:=E_C-\widehat K E_S.
\end{equation}
\end{prop}
\begin{pf}
Since $K^{\rm t}S^{\rm t}=C^{\rm t}$, completing the square in
\eqref{eq:true_covariance} proves \eqref{eq:oracle_completion}. Moreover,
\begin{equation*}
 (\widehat K-K^{\rm t})S^{\rm t}
 =\widehat K S^{\rm t}-C^{\rm t}
 =E_C-\widehat K E_S=\Gamma_{\widehat K}.
\end{equation*}
Substitution into \eqref{eq:oracle_completion} proves
\eqref{eq:moment_regret_identity}. \hfill$\square$
\end{pf}
The covariance penalty is therefore quadratic in the effective error of the
moment pair.

The conventional covariance report is
\begin{equation}\label{eq:conventional_report}
 P_{\rm conv}=P_{\rm rc}(\theta_0;K_0)
 =P-C_0S_0^{-1}C_0^\top.
\end{equation}
Thus one local moment pair selects the correction and also determines the
reported covariance at its estimated minimum. If that pair is inaccurate,
$K_0$ defined in \eqref{eq:center_pair_gain} need not minimize \eqref{eq:true_covariance}, and $P_{\rm conv}$ can
assign excessive confidence to the correction. Figure
\ref{fig:probe_illustration} illustrates this single-center failure mode.

\begin{figure}[htbp]
\centering
\includegraphics[width=0.7\columnwidth]{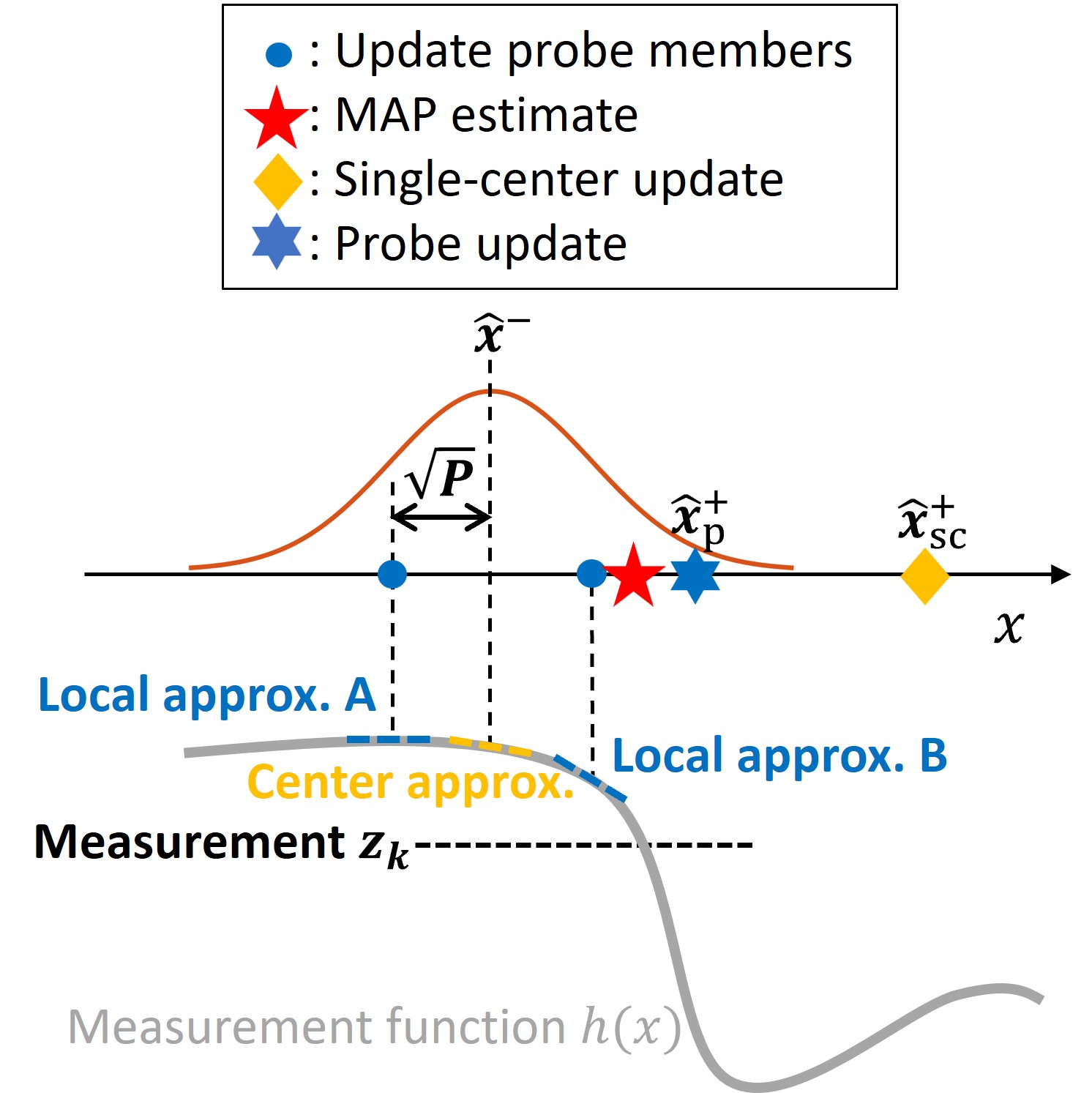}
\caption{Illustration of the update probe in a one-dimensional EKF update.
The predictive density has mean $\hat x^-$ and variance $P$, so the probe
members are $\hat x^-\!\pm\!\sqrt P$. Because $h(x)$ is strongly nonlinear across
this region, the single-center update moves the estimate to $\hat x_{\rm sc}^+$,
away from the maximum a posteriori (MAP) estimate. Averaging the two probe
moment pairs gives $\hat x_{\rm p}^+$, which is much closer to the MAP estimate
in this example.}
\label{fig:probe_illustration}
\end{figure}

\subsection{Update probe and probe-mean gain}\label{sec:update_probe}

The failure depicted in Fig.~\ref{fig:probe_illustration} is caused by
trusting one local approximation at the predicted state as if it were
accurate across the whole uncertainty region. A natural improvement is
therefore to select a Kalman gain that performs well on average across local
approximations evaluated throughout that region. This motivates the update
probe.

The update probe represents a collection of possible states around the
predicted mean. To make the point set representative and computationally
efficient, we use the smallest equally weighted set with mean $\hat x^-$
and covariance $P$. This requires $n+1$ points because $P$ has full rank.
For the one-dimensional example in Fig.~\ref{fig:probe_illustration},
the two members are $\hat x^-\pm\sqrt P$.
For $n$ dimensions, let $P=LL^\top$ be the Cholesky decomposition. Choose $n+1$
regular-simplex vertices $s_i\in\mathbb R^n$ satisfying
\begin{equation}\label{eq:simplex_moments}
 \frac{1}{n+1}\sum_{i=1}^{n+1}s_i=0,
\end{equation}
\begin{equation}\label{eq:simplex_covariance}
 \frac{1}{n+1}\sum_{i=1}^{n+1}s_is_i^\top=I_n.
\end{equation}
For $n=2$, they form an equilateral triangle. In higher dimensions, they are the vertices of a regular simplex, each with Euclidean norm $\sqrt n$. The update-probe members are
\begin{equation}\label{eq:probe_points}
 \chi_i=\hat x^-+Ls_i.
\end{equation}
Their equally weighted displacement covariance is exactly $P$, so the probe spread is
fixed by the covariance carried by the filter and introduces no free scaling
parameter. 

After constructing the probe set, the filter repeats the moment
approximation at each member:
\begin{equation}\label{eq:update_probe_pairs}
 (C_i,S_i):=\Theta_P(\chi_i),
\end{equation}
for all $i=1,\ldots,n+1$. We select a gain that minimizes the average reported error covariance
\begin{equation}\label{eq:average_report}
 \overline P_{\rm rc}(K)
 :=\frac{1}{n+1}\sum_iP_{\rm rc}((C_i,S_i);K).
\end{equation}
Proposition~\ref{prop:consensus} shows that the minimizing gain is
\begin{equation}\label{eq:probe_mean_gain}
 K_{\rm pm}:=\overline C\,\overline S^{-1},
\end{equation}
with
\begin{equation}\label{eq:average_pair}
 \overline C:=\frac{1}{n+1}\sum_i C_i,\qquad
 \overline S:=\frac{1}{n+1}\sum_i S_i.
\end{equation}

\begin{prop}[probe-mean optimality]\label{prop:consensus}
For every $K$,
\begin{equation}\label{eq:consensus_completion}
 \overline P_{\rm rc}(K)-\overline P_{\rm rc}(K_{\rm pm})
 =(K-K_{\rm pm})\overline S(K-K_{\rm pm})^\top\succeq0.
\end{equation}
In particular,
\begin{equation}\label{eq:center_vs_probe}
 \overline P_{\rm rc}(K_0)-\overline P_{\rm rc}(K_{\rm pm})
 =(K_0-K_{\rm pm})\overline S(K_0-K_{\rm pm})^\top\succeq0.
\end{equation}
\end{prop}
\begin{pf}
Affineness of \eqref{eq:generalized_joseph} gives
\begin{equation*}
 \overline P_{\rm rc}(K)
 =P-K\overline C^\top-\overline C K^\top+K\overline S K^\top.
\end{equation*}
Completing the square with $K_{\rm pm}\overline S=\overline C$ proves
\eqref{eq:consensus_completion}. Setting $K=K_0$ proves
\eqref{eq:center_vs_probe}. \hfill$\square$
\end{pf}
Thus $K_{\rm pm}$ is the unique common gain that minimizes the average
reported covariance in every state direction. Equation \eqref{eq:center_vs_probe} gives the excess average local covariance
incurred by the center gain. We apply the probe gain to the same innovation
as the underlying filter, using its predicted measurement $\hat z$
at $\hat x^-$.

Proposition~\ref{prop:consensus} establishes optimality for the average
local covariance report. To understand how probing accounts for
differences among the local models, we separate each local innovation
covariance into a contribution from the local slope and a residual.
Define
\begin{equation}\label{eq:implied_slope}
 A_i:=C_i^\top P^{-1},\qquad
 \Omega_i:=S_i-R-A_iPA_i^\top\succeq0,
\end{equation}
where $A_i$ is the slope encoded by the cross-covariance and $\Omega_i$
contains the measurement uncertainty left after removing the linear
contribution and the noise covariance $R$. Let $\overline A$ and
$\overline\Omega$ denote their arithmetic means, and define the
slope-disagreement covariance
\begin{equation}\label{eq:slope_disagreement}
 \Sigma_A:=\frac{1}{n+1}\sum_i
 (A_i-\overline A)P(A_i-\overline A)^\top\succeq0.
\end{equation}
Since $C_i=PA_i^\top$ and
$S_i=A_iPA_i^\top+R+\Omega_i$, averaging gives
\begin{equation}\label{eq:average_slope_form}
 \overline C=P\overline A^\top,\qquad
 \overline S=\overline A P\overline A^\top
 +R+\overline\Omega+\Sigma_A.
\end{equation}
Thus, averaging the moment pairs incorporates variation among the local
slopes through the additional innovation-covariance term $\Sigma_A$,
which vanishes when the slopes agree. We next establish conditions under
which the resulting probe report is conservative and the probe gain
prevents the true error covariance from increasing.

Quadratic measurements provide a setting in which this effect can be
calculated explicitly. Their slopes vary linearly with the state,
while their curvature is constant. We first take the predicted covariance
as exact to isolate the effect of the moment approximation.
Let $\mu$ be the predictive mean and set $e=x-\mu$. Let $h_a:\mathbb{R}^n\to\mathbb{R}$ denote the $a$th component of the measurement function $h$, for $a=1,\ldots,p$. Suppose
\begin{equation}\label{eq:quadratic_measurement}
 h_a(\mu+e)=h_a(\mu)+g_a^\top e+\tfrac12e^\top H_ae,
 \quad a=1,\ldots,p,
\end{equation}
where $g_a=\nabla h_a(\mu)$ and $H_a=\nabla^2h_a(\mu)$ is symmetric.
Let $G\in\mathbb R^{p\times n}$ have rows $g_a^\top$.
The linear term contributes $B:=GPG^\top$ to the measurement covariance.
To quantify the contribution of curvature, define
\begin{equation}\label{eq:quadratic_residual}
 (\Omega_q)_{ab}:=\tfrac12\operatorname{tr}(H_aPH_bP).
\end{equation}
Under the prediction model in Proposition~\ref{prop:gaussian_quadratic},
$\Omega_q$ is the covariance contributed by the quadratic terms. The following proposition establishes covariance guarantees without requiring the local moment approximation to be exact.

\begin{samepage}
\begin{prop}[quadratic covariance bounds]\label{prop:gaussian_quadratic}
Assume $x-\mu\sim\mathcal N(0,P)$ and $\hat x^-=\mu$, so that
$P=P^{\rm t}$. Suppose that, for every displacement $d$, the local
moment map has the form
\begin{equation}\label{eq:quadratic_map_form}
\begin{aligned}
 C_P(\mu+d)&=PA(d)^\top,\\
 S_P(\mu+d)&=A(d)PA(d)^\top+R+\Omega_{\rm map},
\end{aligned}
\end{equation}
where the $a$th row of $A(d)$ is $g_a^\top+d^\top H_a$ and
$\Omega_{\rm map}\succeq0$ is independent of $d$. Here $\Omega_{\rm map}$ is the residual measurement covariance supplied by the chosen moment approximation after removing the local linear contribution $A(d)PA(d)^\top$ and measurement noise covariance $R$. Then, for every fixed gain $K$,
\begin{equation}\label{eq:quadratic_conservative_report}
 P_{\rm rc}((\overline C,\overline S);K)-P_+^{\rm t}(K)
 =K(\Omega_q+\Omega_{\rm map})K^\top\succeq0.
\end{equation}
At the probe-mean gain,
\begin{equation}\label{eq:quadratic_true_nonincrease}
\begin{aligned}
 P_+^{\rm t}(K_{\rm pm})
 &\preceq P_{\rm rc}((\overline C,\overline S);K_{\rm pm})\\
 &=P-C^{\rm t}\overline S^{-1}C^{{\rm t}\top}
 \preceq P.
\end{aligned}
\end{equation}
\end{prop}
\end{samepage}

\begin{pf}
The centered Gaussian error has zero third moments, giving
$C^{\rm t}=PG^\top$. Isserlis' formula shows that the quadratic
terms have covariance $\Omega_q$, as defined in
\eqref{eq:quadratic_residual}, and hence
$S^{\rm t}=B+R+\Omega_q$.

For $d_i=Ls_i$, the simplex mean yields $\overline A=G$.
Moreover,
\begin{equation*}
\begin{aligned}
 (\Sigma_A)_{ab}
 &=\frac{1}{n+1}\sum_i d_i^\top H_aPH_bd_i\\
 &=\operatorname{tr}\!\left(
 H_aPH_b\frac{1}{n+1}\sum_i d_id_i^\top\right)\\
 &=\operatorname{tr}(H_aPH_bP)
 =2(\Omega_q)_{ab}.
\end{aligned}
\end{equation*}
Substituting $\overline A=G$ and $\Sigma_A=2\Omega_q$ into
\eqref{eq:average_slope_form}, and evaluating
\eqref{eq:quadratic_map_form} at $d=0$, gives
\begin{equation}\label{eq:quadratic_center_and_probe}
\begin{aligned}
 C_0&=\overline C=C^{\rm t},\\
 S_0&=S^{\rm t}+\Omega_{\rm map}-\Omega_q,\\
 \overline S&=S^{\rm t}+\Omega_{\rm map}+\Omega_q
 \succeq S^{\rm t}.
\end{aligned}
\end{equation}
Subtracting \eqref{eq:true_covariance} from
\eqref{eq:generalized_joseph} therefore proves
\eqref{eq:quadratic_conservative_report}.

At $K_{\rm pm}=C^{\rm t}\overline S^{-1}$, the generalized
Joseph report reduces to
$P-C^{\rm t}\overline S^{-1}C^{{\rm t}\top}$, which proves
\eqref{eq:quadratic_true_nonincrease}. \hfill$\square$
\end{pf}
Under the stated assumptions, Proposition~\ref{prop:gaussian_quadratic} shows that the probe-mean update cannot increase the true estimation-error covariance, $P_+^{\rm t}(K_{\rm pm})\preceq P$, and that its reported covariance $P_{\rm rc}((\overline C,\overline S);K_{\rm pm})$ upper-bounds $P_+^{\rm t}(K_{\rm pm})$ without exceeding the predicted covariance $P$. The center innovation covariance upper-bounds the true innovation
covariance precisely when
$\Omega_{\rm map}\succeq\Omega_q$, since
$S_0-S^{\rm t}=\Omega_{\rm map}-\Omega_q$.
Among the four examined implementations, only EKF2 satisfies this condition
throughout the stated class, with $\Omega_{\rm map}=\Omega_q$.
The EKF omits this residual. The UKF and CKF residuals depend on their point
geometry and need not dominate $\Omega_q$.

The same distinction matters for the correction. The probe gain satisfies
\eqref{eq:quadratic_true_nonincrease}, whereas the center gain gives
\begin{equation}\label{eq:center_true_change}
\begin{aligned}
 P_+^{\rm t}(K_0)-P
 &=K_0(S^{\rm t}-2S_0)K_0^\top\\
 &=K_0(\Omega_q-2\Omega_{\rm map}-B-R)K_0^\top,
\end{aligned}
\end{equation}
which can be positive along some state directions. Thus, under the
assumptions of Proposition~\ref{prop:gaussian_quadratic}, probing gives a
conservative covariance report and prevents an increase in the true error
covariance. The single-center update does not generally provide both
guarantees.

Proposition~\ref{prop:gaussian_quadratic} assumes an exact predicted
covariance. In recursive filtering, the carried covariance can understate
prediction uncertainty even when a moment calculation is exact for its
inputs. We next examine when the additional curvature term introduced by
probing improves the gain in this situation. We compare the total normalized
true error variance $\operatorname{tr}(P^{-1}P_+^{\rm t}(K))$.
It sums the error variances after normalizing the state by the carried
covariance $P$ and is invariant under invertible linear changes of state
coordinates. We first set $R=0$ to isolate the effect of measurement curvature.

\begin{prop}[benefit under overconfidence]
\label{prop:understated_scale}
Consider the quadratic measurement model \eqref{eq:quadratic_measurement}
and moment-map form \eqref{eq:quadratic_map_form}.
Let the true prediction error $x-\hat x^-$ be Gaussian with mean zero
and covariance $\upsilon P$, where $\upsilon>0$, while the filter carries
$\hat x^-=\mu$ and covariance $P\succ0$. Assume $R=0$, $B=GPG^\top\succ0$, and,
for some $\delta\ge0$,
\begin{equation}\label{eq:scale_residual_bound}
 0\preceq\Omega_{\rm map}\preceq\delta\Omega_q.
\end{equation}
The center gain and probe-mean gain are
\begin{equation}\label{eq:scale_center_probe_gains}
\begin{aligned}
 K_0
 &=PG^\top(B+\Omega_{\rm map})^{-1},\\
 K_{\rm pm}
 &=PG^\top(B+\Omega_{\rm map}+2\Omega_q)^{-1}.
\end{aligned}
\end{equation}
If $\upsilon\ge1+\delta$, their true error covariances satisfy
\begin{equation}\label{eq:scale_same_map}
 \operatorname{tr}\!\left(P^{-1}P_+^{\rm t}(K_{\rm pm})\right)
 \le
 \operatorname{tr}\!\left(P^{-1}P_+^{\rm t}(K_0)\right).
\end{equation}
The inequality is strict when $\Omega_q\ne0$.
\end{prop}
\begin{pf}
We compare the two gains directly after normalizing by the center
innovation covariance. Under the stated quadratic Gaussian model,
\begin{equation}\label{eq:scaled_true_pair}
 C^{\rm t}=\upsilon PG^\top,\qquad
 S^{\rm t}=\upsilon(B+\upsilon\Omega_q).
\end{equation}
Set $S_0:=B+\Omega_{\rm map}\succ0$ and define
\begin{equation}\label{eq:scale_center_normalization}
 E:=S_0^{-1/2}BS_0^{-1/2},\qquad
 T:=S_0^{-1/2}\Omega_qS_0^{-1/2},
\end{equation}
where $S_0^{-1/2}$ is the symmetric inverse square root of $S_0$.
The residual bound becomes
\begin{equation}\label{eq:scale_normalized_bounds}
 0\prec E\preceq I,\qquad
 0\preceq I-E\preceq\delta T.
\end{equation}
In these coordinates, the probe gain differs from the center gain
through the factor $J:=(I+2T)^{-1}$:
\begin{equation}\label{eq:scale_probe_factor}
\begin{aligned}
 K_0&=PG^\top S_0^{-1/2}I S_0^{-1/2},\\
 K_{\rm pm}&=PG^\top S_0^{-1/2}J S_0^{-1/2}.
\end{aligned}
\end{equation}
It remains to show that inserting $J$ does not increase the normalized
true error variance. Define its change, divided by $\upsilon>0$, as
\begin{equation}\label{eq:scale_risk_difference}
 \Delta:=\frac{1}{\upsilon}\operatorname{tr}\!\left[
 P^{-1}\bigl(P_+^{\rm t}(K_{\rm pm})-P_+^{\rm t}(K_0)\bigr)
 \right].
\end{equation}
To evaluate $\Delta$, first note that $S_0^{-1/2}GPG^\top S_0^{-1/2}=E$ and $S_0^{-1/2}S^{\rm t}S_0^{-1/2}
 =\upsilon(E+\upsilon T)$.
Substituting $K_{\rm pm}$ into \eqref{eq:true_covariance}
and using cyclicity of the trace therefore gives
\begin{equation}\label{eq:scale_probe_risk}
\begin{aligned}
 \frac{1}{\upsilon}
 \operatorname{tr}\!\left[P^{-1}P_+^{\rm t}(K_{\rm pm})\right]
 ={}&n-2\operatorname{tr}(EJ)\\
 &+\operatorname{tr}\!\left[EJ(E+\upsilon T)J\right].
\end{aligned}
\end{equation}
For the center gain, the same calculation with $J=I$ yields
\begin{equation}\label{eq:scale_center_risk}
\begin{aligned}
 \frac{1}{\upsilon}
 \operatorname{tr}\!\left[P^{-1}P_+^{\rm t}(K_0)\right]
 ={}&n-2\operatorname{tr}(E)\\
 &+\operatorname{tr}\!\left[E(E+\upsilon T)\right].
\end{aligned}
\end{equation}
Subtracting \eqref{eq:scale_center_risk} from
\eqref{eq:scale_probe_risk} cancels the prior contribution $n$.
Moreover, since $J$ is a function of $T$, we have $JT=TJ$ and hence
\begin{equation}\label{eq:scale_quadratic_trace_expansion}
\begin{aligned}
 \operatorname{tr}\!\left[EJ(E+\upsilon T)J\right]
 &=\operatorname{tr}(EJEJ)
   +\upsilon\operatorname{tr}(EJTJ)\\
 &=\operatorname{tr}(EJEJ)
   +\upsilon\operatorname{tr}(ETJ^2).
\end{aligned}
\end{equation}
Combining these expressions gives
\begin{equation}\label{eq:scale_risk_difference_trace}
\begin{aligned}
 \Delta={}&2\operatorname{tr}[E(I-J)]
 +\operatorname{tr}(EJEJ-E^2)\\
 &+\upsilon\operatorname{tr}[ET(J^2-I)].
\end{aligned}
\end{equation}
Since $J=(I+2T)^{-1}$, an orthonormal eigenbasis of $T$ also
diagonalizes $J$. Choose an orthogonal matrix $U$ such that
\begin{equation}\label{eq:scale_diagonal_coordinates}
\begin{aligned}
 U^\top T U&=\operatorname{diag}(t_1,\ldots,t_p),\\
 U^\top J U&=\operatorname{diag}(j_1,\ldots,j_p),
\end{aligned}
\end{equation}
where $t_i\ge0$ and $j_i=(1+2t_i)^{-1}\in(0,1]$.
Write $U^\top E U=(e_{ij})$.
Taking diagonal entries in \eqref{eq:scale_normalized_bounds}
in these coordinates gives
\begin{equation}\label{eq:scale_diagonal_entry_bounds}
 e_{ii}>0,\qquad
 0\le1-e_{ii}\le\delta t_i.
\end{equation}

The trace is unchanged by this orthogonal change of coordinates.
Because $J$ and $T$ are diagonal in these coordinates,
\begin{equation}\label{eq:scale_linear_trace_sums}
\begin{aligned}
 \operatorname{tr}[E(I-J)]
 &=\sum_i e_{ii}(1-j_i),\\
 \operatorname{tr}[ET(J^2-I)]
 &=-\sum_i e_{ii}t_i(1-j_i^2).
\end{aligned}
\end{equation}
For the remaining trace, symmetry of $E$ gives $e_{ji}=e_{ij}$,
so
\begin{equation}\label{eq:scale_quadratic_trace_sums}
\begin{aligned}
 \operatorname{tr}(EJEJ)
 &=\sum_{i,j}e_{ij}j_j e_{ji}j_i
 =\sum_{i,j}e_{ij}^2j_i j_j,\\
 \operatorname{tr}(E^2)
 &=\sum_{i,j}e_{ij}e_{ji}
 =\sum_{i,j}e_{ij}^2.
\end{aligned}
\end{equation}
Substituting these identities into
\eqref{eq:scale_risk_difference_trace} yields
\begin{equation}\label{eq:scale_risk_difference_sum}
\begin{aligned}
 \Delta=\sum_i e_{ii}
 \bigl[2(1-j_i)-\upsilon t_i(1-j_i^2)\bigr]-\sum_{i,j}e_{ij}^2(1-j_i j_j).
\end{aligned}
\end{equation}
Since $0<j_i\le1$, the terms with $i\ne j$ enter with a
nonpositive sign. Omitting these terms and using
$2(1-j_i)=2(1-j_i^2)/(1+j_i)$ gives the upper bound
\begin{equation}\label{eq:scale_diagonal_bound}
 \Delta\le\sum_i e_{ii}(1-j_i^2)
 \left[\frac{2}{1+j_i}-e_{ii}-\upsilon t_i\right].
\end{equation}
Since $2/(1+j_i)=1+t_i/(1+t_i)$, the bounds on $e_{ii}$ and
$\upsilon\ge\delta+1$ imply
\begin{equation}\label{eq:scale_component_bound}
\begin{aligned}
 \frac{2}{1+j_i}-e_{ii}-\upsilon t_i
 &\le t_i\left[\frac{1}{1+t_i}+\delta-\upsilon\right]\\
 &\le-\frac{t_i^2}{1+t_i}.
\end{aligned}
\end{equation}
Consequently,
\begin{equation}\label{eq:scale_direct_comparison}
 \Delta\le-\sum_i e_{ii}(1-j_i^2)
 \frac{t_i^2}{1+t_i}\le0.
\end{equation}
If $\Omega_q\ne0$, then some $t_i>0$ and $j_i<1$.
Since $e_{ii}>0$, the final bound is strictly negative.
If $\Omega_q=0$, then $J=I$ and the two gains
coincide. This proves both claims. \hfill$\square$
\end{pf}

Proposition~\ref{prop:understated_scale} compares each local moment
map with its own probed version. Because \eqref{eq:scale_residual_bound} is a
matrix inequality, $\Omega_{\rm map}$ need not be a scalar multiple of $\Omega_q$.
Once the true prediction-error covariance $\upsilon P$ is large enough relative
to the carried covariance $P$ that $\upsilon\ge1+\delta$, the additional term
$2\Omega_q$ in \eqref{eq:scale_center_probe_gains} gives a gain whose total
normalized true error variance is no larger than that of $K_0$, and strictly
smaller when $\Omega_q\ne0$.

For an exact local moment rule, $\Omega_{\rm map}=\Omega_q$.
Taking $\delta=1$ shows that probing this rule improves on its center
gain when $\upsilon\ge2$ and $\Omega_q\ne0$.
Thus, even exact local moments can benefit from probing when they are
computed using an understated prediction covariance.
The improvement comes from the slope variation across probe members,
which changes the gain beyond the moment calculation at a single center.

For fixed model matrices and $\Omega_q\ne0$, the strict advantage also
persists for sufficiently small $R\succeq0$ by continuity. Here, smallness can
be measured by $\|B^{-1/2}RB^{-1/2}\|_2$, where $\|\cdot\|_2$ is the
spectral norm.

Altogether, Propositions~\ref{prop:consensus}--\ref{prop:understated_scale}
explain how probe averaging affects the state update. It selects the best
common gain for the moment pairs across the probe. Under the assumptions of
Proposition~\ref{prop:gaussian_quadratic}, the added slope variation provides a conservative
covariance report and prevents an increase in the true error covariance.
It also improves the gain in normalized true error variance when prediction
uncertainty is sufficiently understated. To carry the resulting uncertainty into the next prediction, the filter still needs a
covariance report at the updated estimate. This is the role of recalibration.

\subsection{Covariance recalibration and back-out}\label{sec:recalibration}

The conventional report \eqref{eq:conventional_report} is the minimum of the
generalized Joseph report built from the same approximate pair that selected the gain. In
Fig.~\ref{fig:probe_illustration}, the conventional EKF correction moves
away from the MAP estimate, yet its reported covariance decreases.
Theorem~1 of \cite{parent} formalizes this tendency toward overconfidence:
under conditionally unbiased cross- and innovation-covariance approximations,
an exact predicted covariance makes the conventional minimum-form report no
larger than the candidate-update covariance on average.

Covariance recalibration \cite{parent} repeats the moment calculation at the
updated state using the predicted covariance. It holds the applied gain $K$
fixed and inserts the new pair into the generalized Joseph report.
Single-point recalibration is
\begin{equation}\label{eq:single_recal_pair}
 \check\theta_0:=\Theta_P(\hat x^+),\qquad
 P_{\rm one}:=P_{\rm rc}(\check\theta_0;K).
\end{equation}
In Fig.~\ref{fig:probe_illustration}, this means relinearizing the
measurement function at $\hat x_{\rm sc}^+$ to reassess the uncertainty
after the correction.

For the probe framework, recalibration recenters the probe at the updated
state and repeats the probe-averaged moment calculation using the predicted
covariance. This recentered probe is called the \emph{recalibration probe}. It retains the spread $P$ because the report
estimates how the applied gain transforms the predicted uncertainty. The same
factor $L$ and simplex vertices $s_i$ are used in both probe stages.
With $K=K_{\rm pm}$, define
\begin{equation}\label{eq:recalibration_points}
 \check\chi_i=\hat x^++Ls_i,\qquad
 \check\theta_i=(\check C_i,\check S_i):=\Theta_P(\check\chi_i),
\end{equation}
\begin{equation}\label{eq:recalibration_average}
 \check C_{\rm av}:=\frac{1}{n+1}\sum_i\check C_i,\qquad
 \check S_{\rm av}:=\frac{1}{n+1}\sum_i\check S_i,
\end{equation}
and
\begin{equation}\label{eq:probe_recal_report}
 P_{\rm rep}:=P_{\rm rc}((\check C_{\rm av},\check S_{\rm av});K_{\rm pm}).
\end{equation}
Because \eqref{eq:generalized_joseph} is affine in the moment pair,
\begin{equation}\label{eq:average_of_reports}
 P_{\rm rep}=\frac{1}{n+1}\sum_i
 P_{\rm rc}(\check\theta_i;K_{\rm pm}).
\end{equation}

As illustrated in Fig.~\ref{fig:probe_illustration}, a nonlinear KF update
does not necessarily improve the estimate. Recalibration gives the filter
a way to reassess the uncertainty after the correction. In the scalar case,
$P_{\rm rep}>P$ means that the revised local approximation assigns greater
uncertainty to the corrected estimate. A reasonable response is to withdraw
the update and return to the predicted mean and covariance. More generally,
for a positive definite weight $M$, define
$J_M(A):=\operatorname{tr}(MA)$ as a measure of uncertainty for a covariance
matrix $A$. The back-out step is triggered when
$J_M(P_{\rm rep})>J_M(P)$.

Our earlier work uses $M=I$ \cite{parent}, so the back-out criterion compares
the sums of the state variances before and after the update. This metric is
appropriate when the state components are directly comparable. When they
have different units, adding their variances has no clear interpretation.
We therefore use $M=P^{-1}$ to normalize the comparison by the predicted
uncertainty. In particular,
\begin{equation}\label{eq:normalized_trace}
 J_{P^{-1}}(A)=\operatorname{tr}(P^{-1}A),\qquad J_{P^{-1}}(P)=n.
\end{equation}
This uses the same normalization as Proposition~\ref{prop:understated_scale},
applied here to the reported covariance.
This metric is unchanged by an invertible linear coordinate transformation $x'=Ux$.
Indeed, with $A'=UAU^\top$ and $P'=UPU^\top$, cyclic invariance of the trace
gives
\begin{equation}\label{eq:norm_invariance}
 J_{(P')^{-1}}(A')=\operatorname{tr}((UPU^\top)^{-1}UAU^\top)
 =J_{P^{-1}}(A).
\end{equation}
After conditional back-out, the output state estimate and covariance are
\begin{equation}\label{eq:backout_output}
 (\hat x^{\rm out},P^{\rm out})=
 \begin{cases}
 (\hat x^+,P_{\rm rep}),&J_{P^{-1}}(P_{\rm rep})\le n,\\
 (\hat x^-,P),&J_{P^{-1}}(P_{\rm rep})>n.
 \end{cases}
\end{equation}
Both branches satisfy $J_{P^{-1}}(P^{\rm out})\le n$. A rejected update
returns both the state and covariance to their predicted values.

\subsection{Complete algorithm}\label{sec:algorithm}

Algorithm \ref{alg:main} summarizes the proposed framework. The underlying filter
computes the predicted measurement at $\hat x_{k|k-1}$. Every measurement
moment calculation uses the predicted covariance $P_{k|k-1}$ as its spread.
For a linear system and a moment rule exact on affine maps, the covariance
pair is independent of the evaluation center. Both probes then return the
conventional pair, recalibration leaves the report unchanged, and the
back-out rule accepts the update. The framework therefore reduces to the
linear KF in this special case.

\begin{algorithm}[htbp]
\caption{Probe-set framework for nonlinear Kalman filtering}
\label{alg:main}
\begin{algorithmic}[1]
\State Initialize $\hat x_{0|0}$ and $P_{0|0}$
\For{$k=1,2,\ldots$}
  \State Predict $\hat x_{k|k-1}$ and $P_{k|k-1}$
  \State Compute the predicted measurement $\hat z_{k|k-1}$
  \State Factor $P_{k|k-1}=L_kL_k^\top$ and form the simplex
  \State Evaluate $(C_i^-,S_i^-)$ at $\hat x_{k|k-1}+L_ks_i$
  \State $\overline C_k^-\gets(n+1)^{-1}\sum_iC_i^-$,
         $\overline S_k^-\gets(n+1)^{-1}\sum_iS_i^-$
  \State $K_k\gets\overline C_k^-(\overline S_k^-)^{-1}$
  \State Apply $\hat x_{k|k}\gets\hat x_{k|k-1}
         +K_k(z_k-\hat z_{k|k-1})$
  \State Evaluate $(C_i^+,S_i^+)$ at $\hat x_{k|k}+L_ks_i$
  \State $\overline C_k^+\gets(n+1)^{-1}\sum_iC_i^+$,
         $\overline S_k^+\gets(n+1)^{-1}\sum_iS_i^+$
  \State $P_{k|k}\gets P_{k|k-1}-K_k\overline C_k^{+\top}
         -\overline C_k^+K_k^\top+K_k\overline S_k^+K_k^\top$
  \If{$\operatorname{tr}(P_{k|k-1}^{-1}P_{k|k})>n$}
    \State Back out: revert to $\hat x_{k|k-1}$ and $P_{k|k-1}$
  \EndIf
\EndFor
\end{algorithmic}
\end{algorithm}

\section{Results}\label{sec:results}

\subsection{Simulation protocol}\label{sec:protocol}

The experiments test accuracy, covariance consistency, and computational
cost on two systems: battery state estimation and terrain-referenced navigation. The battery example estimates state of charge (SOC) and state of health
(SOH) from terminal voltage under a known current profile. The terrain example
estimates two-dimensional position from elevation measurements and a known
terrain map. Both simulations contain 60 updates.

To evaluate probing across KF variants, we consider EKF, EKF2, UKF, and CKF
in each application. Table \ref{tab:frameworks} summarizes six configurations
for comparison and ablation. F4, in the row with gray background, is the proposed framework
(Algorithm \ref{alg:main}). F0, F1, and the iterated posterior linearization
filter (IPLF) are three published baselines. Specifically, F0 is the conventional
implementation of each KF variant. F1 uses single-point recalibration with the $M=I$ back-out
rule \cite{parent}, and IPLF iteratively refines its local model within each
update \cite{plf}. These comparisons test the probe framework's added benefit relative to other relevant published ideas. F2 and F3 omit the update and recalibration probes, respectively, while retaining the unit-invariant back-out rule ($M=P^{-1}$). Comparisons among F2--F4 isolate the two probe effects.

\begin{table}[htbp]
\centering
\caption{Principal methods used in the numerical study. Here $P$ is the
predicted covariance and $M$ is the trace weight. The proposed framework is F4.}
\label{tab:frameworks}
\scriptsize
\setlength{\tabcolsep}{2.4pt}
\renewcommand{\arraystretch}{1.16}
\setlength{\aboverulesep}{0pt}
\setlength{\belowrulesep}{0pt}
\vspace{3pt}
\begin{tabular}{l|L{1.55cm}|L{3cm}|L{1.55cm}}
\toprule
Label & Gain & Covariance report & Back-out \\
\midrule
F0 & Center & Center-pair minimum & -- \\
F1 & Center & Single recalibration & $M=I$ \\
F2 & Center & Probe average & $M=P^{-1}$ \\
F3 & Probe mean & Single recalibration & $M=P^{-1}$ \\
\rowcolor{proposedgray}
F4 & Probe mean & Probe average & $M=P^{-1}$ \\
IPLF & Iterated & Standard IPLF & -- \\
\bottomrule
\end{tabular}
\end{table}

The IPLF uses a Kullback--Leibler stopping tolerance $\gamma=10^{-4}$ and at
most 20 iterations per update \cite{plf}. The UKF uses $\alpha=10^{-3}$, $\beta=2$, and
$\kappa=0$ \cite{UKF2}. Appendix \ref{app:filters} gives the implemented moment
formulas, and Appendix \ref{app:systems} specifies both system models.

We use two types of experiments to compare filter performance. The
measurement-noise sweep illustrates sensitivity to measurement noise. It tests
F0--F4 and IPLF at nine noise levels with $10{,}000$ paired Monte Carlo runs
per level. All paired
methods receive identical initial errors and noise realizations. The noise standard deviation $\sigma_z$ ranges from $10^{-6}$ to $10^{-2}$ V
for the battery and from $1$ to $10^3$ m for terrain-referenced navigation. 

The setup-sweep experiment varies the physical conditions and initial uncertainty.
A setup fixes the model parameters, noise levels, and initial-error scales.
Runs within a setup redraw the initial error and noise sequences. We sample
300 setups per application and evaluate all four underlying filters. One
\emph{filter-setup cell} is a physical setup evaluated with one such filter,
giving 1,200 cells per application. Each cell uses 1,000 paired runs of the six
principal configurations and its corresponding higher-order baseline. The varied
setup parameters and associated ranges are given in Appendix \ref{app:systems}.
This study evaluates each framework--variant combination across a broad
range of conditions.

In the measurement-noise-sweep study, the system parameters (except $\sigma_z$) are rounded versions of battery setup 55
and terrain setup 190 in the setup-sweep study. We selected these cases from the sampled population because they clearly separate the performance of the filters.

\subsection{State-estimation accuracy}

\begin{figure*}[htbp]
\centering
\includegraphics[width=0.48\textwidth]{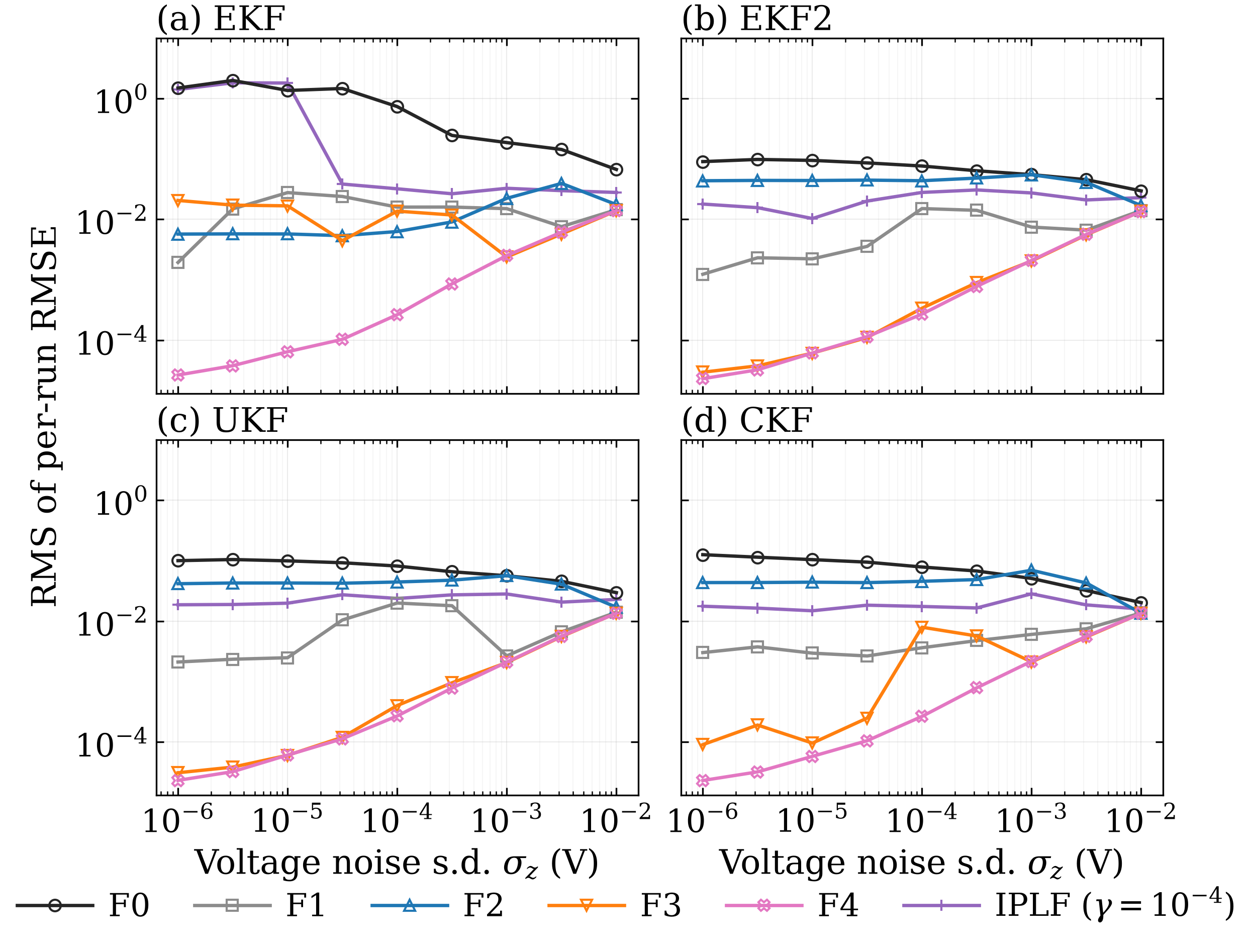}\hfill
\includegraphics[width=0.48\textwidth]{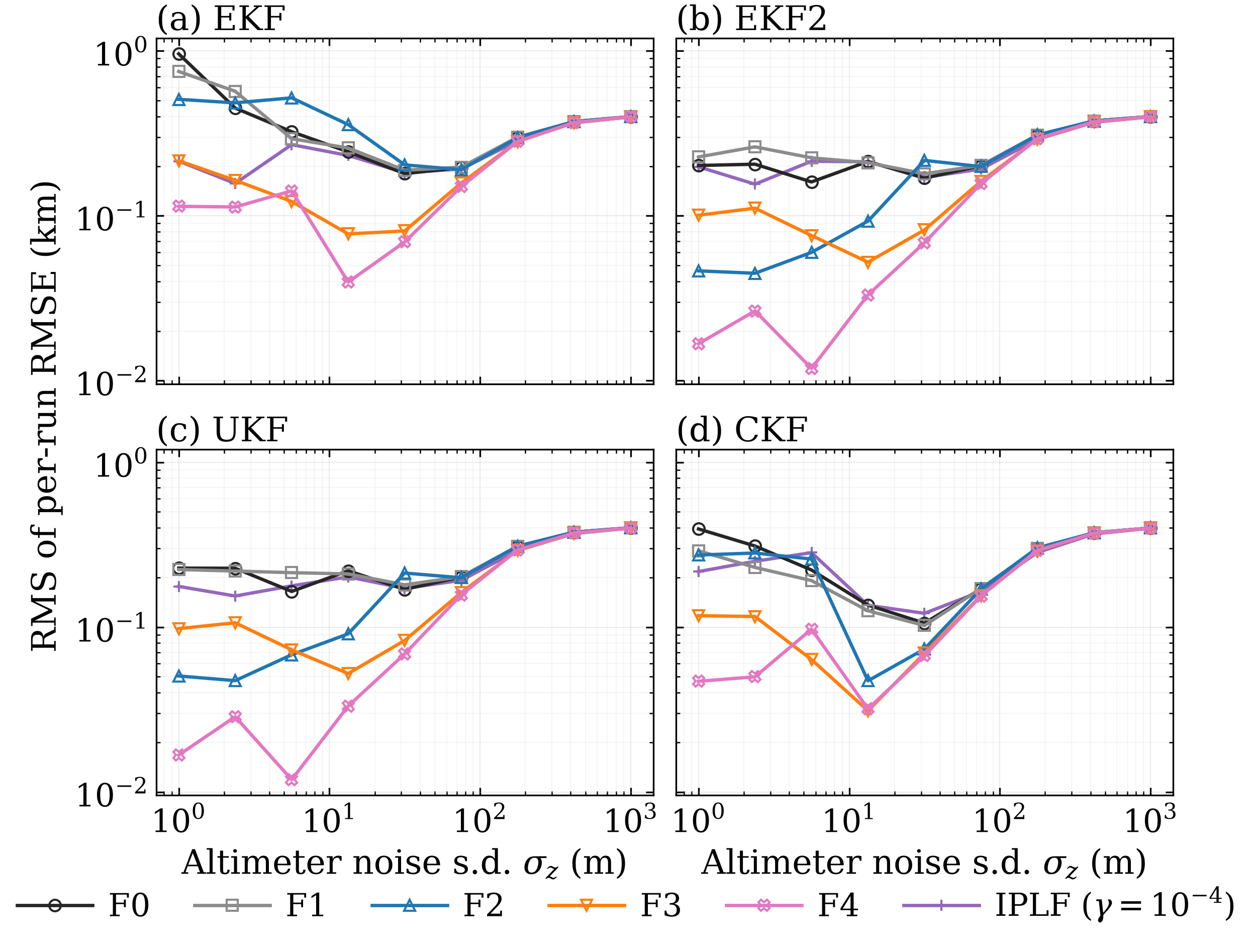}
\caption{RMS of per-run RMSE in the fixed measurement-noise sweeps.
Each per-run value averages the two key-state RMSEs over the final 25\%
of the trajectory. Left: battery state estimation. Right:
terrain-referenced navigation. Each point uses $10{,}000$ paired Monte
Carlo runs.}
\label{fig:rmse_sweeps}
\end{figure*}

\begin{figure*}[htbp]
\centering
\includegraphics[width=0.48\textwidth]{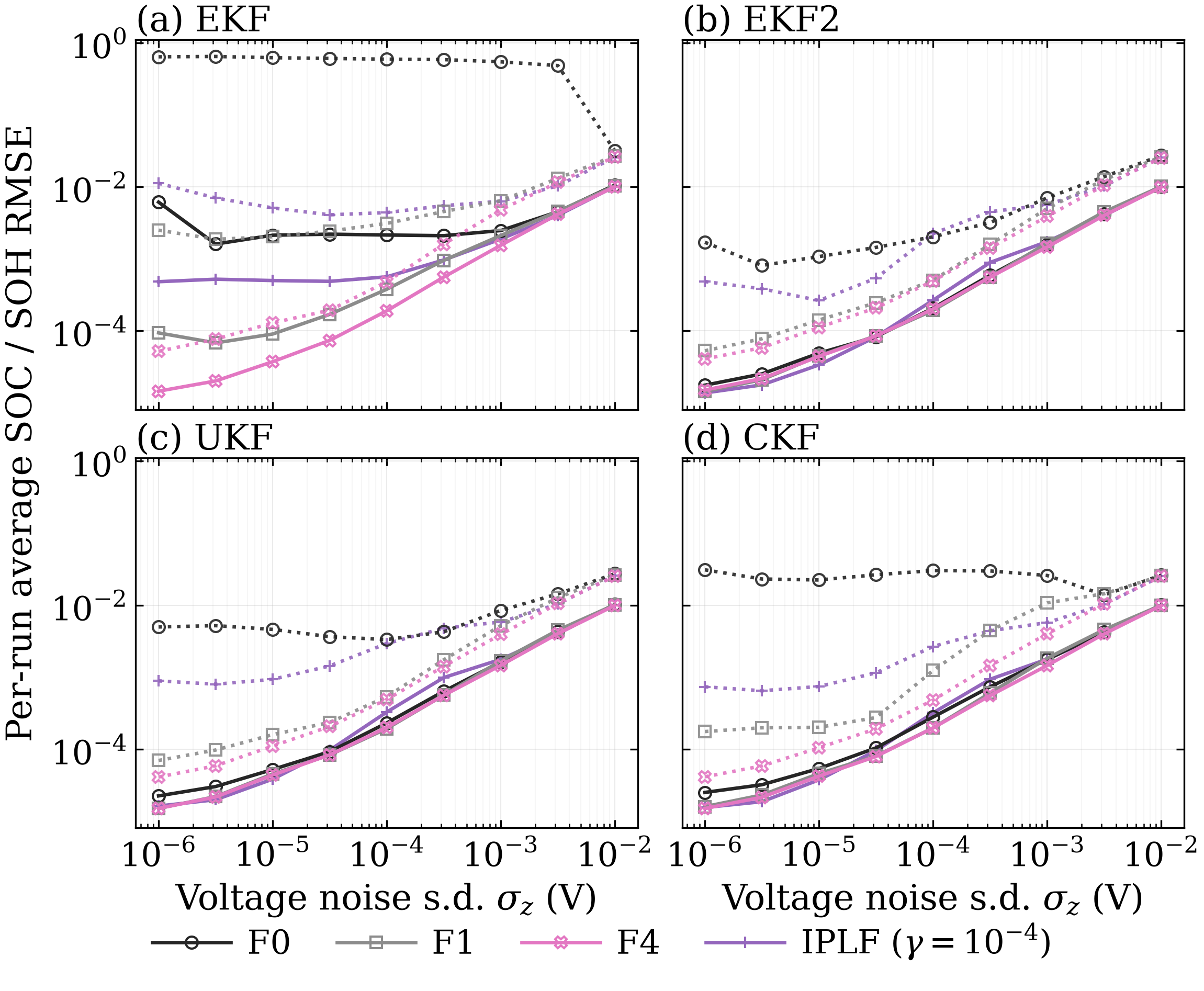}\hfill
\includegraphics[width=0.48\textwidth]{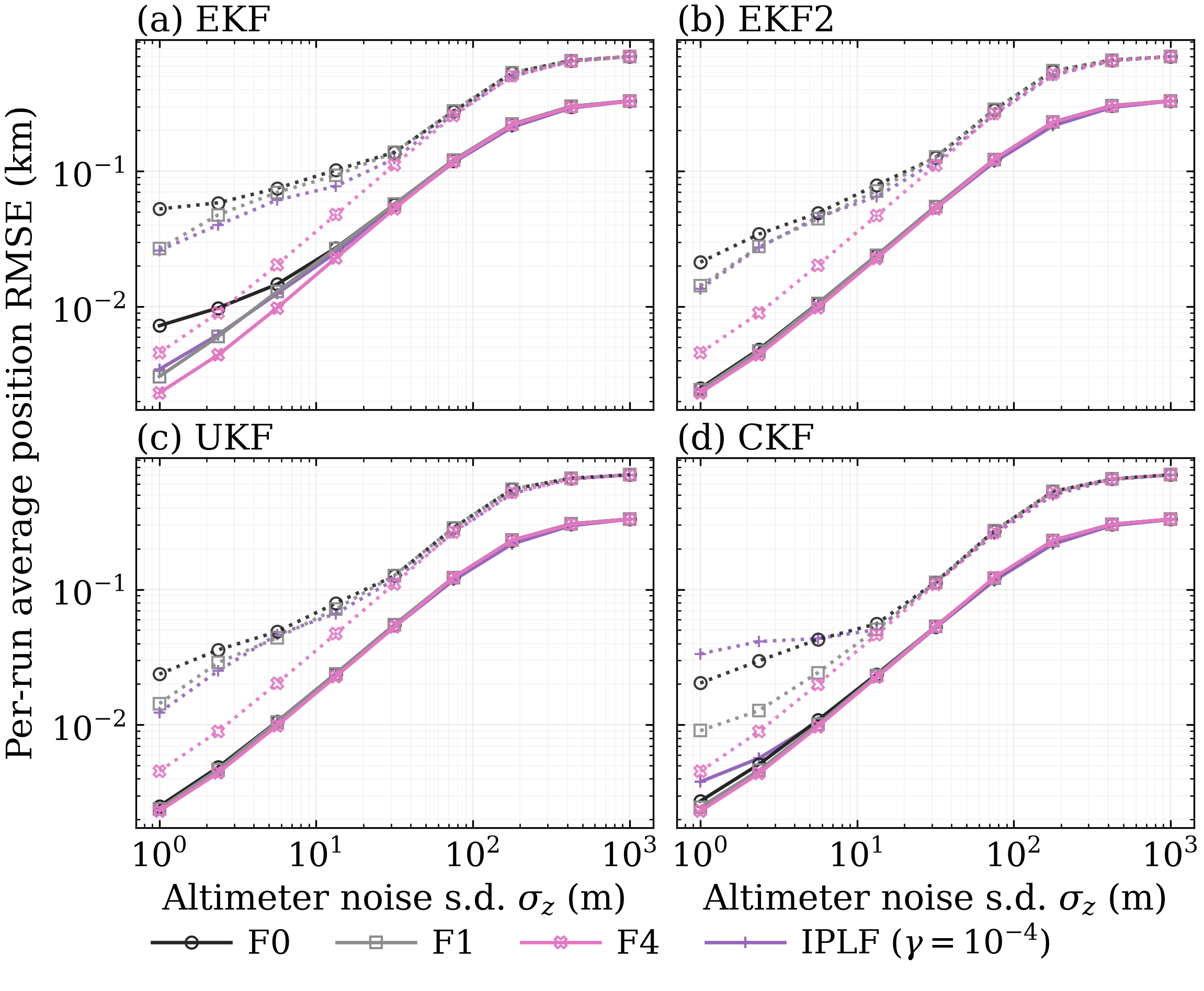}
\caption{Median (solid) and 95th percentile (dotted) of the per-run RMSE over
the final 25\% of each trajectory. Left: battery state estimation. Right:
terrain-referenced navigation. F2 and F3 are omitted for readability.}
\label{fig:rmse_tails}
\end{figure*}

We evaluate state-estimation accuracy using SOC and SOH for the battery,
and both position coordinates for terrain-referenced navigation. For each
key state and each run, we calculate the root mean square error (RMSE)
over the final 25\% of the trajectory (the last 15 steps). This prevents
errors in the first few steps from dominating the accuracy measure.
We then average the two key-state RMSEs to obtain a single per-run RMSE.
For the measurement-noise sweep, Fig. \ref{fig:rmse_sweeps} reports the
root mean square (RMS) of these values across $10{,}000$ Monte Carlo runs,
and Fig. \ref{fig:rmse_tails} reports their median and 95th percentile.
For the setup sweep, we calculate the same three statistics within each
setup. Table \ref{tab:randomized} reports the geometric mean of each
statistic over the 300 setups. Since RMSE can vary greatly across setups,
the geometric mean summarizes performance on a relative scale.

\begin{table*}[htbp]
\centering
\caption{Setup sweep with 300 physical setups per application and $1{,}000$ paired runs per filter-setup cell. Each per-run RMSE averages the two key-state RMSEs. Within each setup, we calculate the RMS, median, and 95th percentile of these values, and ANEES from \eqref{eq:anees}. Each entry is the geometric mean of the corresponding statistic over 300 setups. Battery RMSE is dimensionless. Navigation RMSE is in kilometers. ANEES has nominal value one. The four nonlinear-KF variants are reported separately. Within each variant and application, all values no more than 5\% above the minimum for that metric are bolded. Method F4, in the rows with gray background, is the proposed approach.}
\label{tab:randomized}
\scriptsize
\setlength{\tabcolsep}{2.4pt}
\renewcommand{\arraystretch}{1.16}
\setlength{\aboverulesep}{0pt}
\setlength{\belowrulesep}{0pt}
\vspace{3pt}
\begin{tabular}{llcccc|cccc}
\toprule
\multicolumn{6}{c|}{Battery state estimation}
& \multicolumn{4}{c}{Terrain-referenced navigation}\\
KF & Method & RMS & Median & 95th percentile & ANEES
& RMS & Median & 95th percentile & ANEES\\
\midrule
\textbf{EKF} & F0 & $0.0224$ & $0.00123$ & $0.0191$ & $4.22\mathbin{\times}10^{8}$ & $0.246$ & $0.0279$ & $0.176$ & $84.2$\\
 & F1 & $0.00183$ & $0.000370$ & $0.00177$ & $8.32$ & $0.133$ & $\boldsymbol{0.0266}$ & $0.0919$ & $21.1$\\
 & F2 & $0.00167$ & $0.000374$ & $0.00185$ & $\boldsymbol{1.28}$ & $0.0972$ & $\boldsymbol{0.0269}$ & $0.0769$ & $8.48$\\
 & F3 & $0.00120$ & $0.000311$ & $0.00140$ & $9.01$ & $\boldsymbol{0.0849}$ & $\boldsymbol{0.0263}$ & $0.0744$ & $9.82$\\
 \rowcolor{proposedgray}
 & F4 & $\boldsymbol{0.000734}$ & $\boldsymbol{0.000291}$ & $\boldsymbol{0.00103}$ & $1.99$ & $\boldsymbol{0.0828}$ & $\boldsymbol{0.0262}$ & $\boldsymbol{0.0696}$ & $\boldsymbol{7.04}$\\
 & IPLF & $0.00178$ & $0.000411$ & $0.00184$ & $42.6$ & $0.224$ & $\boldsymbol{0.0272}$ & $0.169$ & $86.2$\\
\midrule
\textbf{EKF2} & F0 & $0.00189$ & $\boldsymbol{0.000293}$ & $0.00144$ & $41.8$ & $0.103$ & $\boldsymbol{0.0261}$ & $0.0748$ & $15.0$\\
 & F1 & $0.000809$ & $\boldsymbol{0.000296}$ & $0.000939$ & $1.87$ & $0.0826$ & $\boldsymbol{0.0261}$ & $0.0682$ & $8.18$\\
 & F2 & $0.00109$ & $0.000371$ & $0.00112$ & $\boldsymbol{1.06}$ & $0.0810$ & $\boldsymbol{0.0271}$ & $0.0648$ & $5.53$\\
 & F3 & $0.000696$ & $0.000314$ & $0.000959$ & $1.72$ & $\boldsymbol{0.0678}$ & $\boldsymbol{0.0261}$ & $0.0639$ & $5.84$\\
 \rowcolor{proposedgray}
 & F4 & $\boldsymbol{0.000574}$ & $\boldsymbol{0.000297}$ & $\boldsymbol{0.000832}$ & $1.27$ & $\boldsymbol{0.0657}$ & $\boldsymbol{0.0261}$ & $\boldsymbol{0.0603}$ & $\boldsymbol{4.23}$\\
 & IPLF & $0.00141$ & $\boldsymbol{0.000304}$ & $0.00119$ & $14.9$ & $0.113$ & $\boldsymbol{0.0262}$ & $0.0857$ & $22.1$\\
\midrule
\textbf{UKF} & F0 & $0.00276$ & $\boldsymbol{0.000300}$ & $0.00177$ & $130$ & $0.108$ & $\boldsymbol{0.0261}$ & $0.0781$ & $16.4$\\
 & F1 & $0.000903$ & $\boldsymbol{0.000297}$ & $0.000972$ & $2.42$ & $0.0839$ & $\boldsymbol{0.0261}$ & $0.0688$ & $8.79$\\
 & F2 & $0.00110$ & $0.000373$ & $0.00113$ & $\boldsymbol{1.07}$ & $0.0790$ & $\boldsymbol{0.0271}$ & $0.0650$ & $5.36$\\
 & F3 & $0.000764$ & $0.000315$ & $0.000964$ & $1.97$ & $0.0705$ & $\boldsymbol{0.0262}$ & $0.0644$ & $6.19$\\
 \rowcolor{proposedgray}
 & F4 & $\boldsymbol{0.000596}$ & $\boldsymbol{0.000296}$ & $\boldsymbol{0.000832}$ & $1.28$ & $\boldsymbol{0.0634}$ & $\boldsymbol{0.0261}$ & $\boldsymbol{0.0603}$ & $\boldsymbol{4.05}$\\
 & IPLF & $0.00149$ & $\boldsymbol{0.000309}$ & $0.00125$ & $17.5$ & $0.118$ & $\boldsymbol{0.0262}$ & $0.0861$ & $22.9$\\
\midrule
\textbf{CKF} & F0 & $0.00191$ & $\boldsymbol{0.000305}$ & $0.00145$ & $46.9$ & $0.141$ & $\boldsymbol{0.0267}$ & $0.0967$ & $27.4$\\
 & F1 & $\boldsymbol{0.000722}$ & $\boldsymbol{0.000297}$ & $\boldsymbol{0.000988}$ & $2.38$ & $0.0916$ & $\boldsymbol{0.0263}$ & $0.0705$ & $10.1$\\
 & F2 & $0.00121$ & $0.000482$ & $0.00139$ & $\boldsymbol{1.12}$ & $0.0840$ & $\boldsymbol{0.0268}$ & $0.0665$ & $6.22$\\
 & F3 & $0.000730$ & $0.000318$ & $0.00112$ & $2.42$ & $\boldsymbol{0.0715}$ & $\boldsymbol{0.0262}$ & $0.0679$ & $6.80$\\
 \rowcolor{proposedgray}
 & F4 & $\boldsymbol{0.000689}$ & $\boldsymbol{0.000311}$ & $\boldsymbol{0.000993}$ & $1.60$ & $\boldsymbol{0.0711}$ & $\boldsymbol{0.0262}$ & $\boldsymbol{0.0632}$ & $\boldsymbol{5.38}$\\
 & IPLF & $0.00153$ & $0.000318$ & $0.00129$ & $18.6$ & $0.133$ & $\boldsymbol{0.0265}$ & $0.101$ & $38.0$\\
\bottomrule
\end{tabular}
\end{table*}
 
Figure \ref{fig:rmse_sweeps} shows that F4's advantage is largest at low
measurement noise, where the measurement is more informative. Compared
with the three published baselines (F0, F1, and IPLF), F4 can reduce RMS
by several orders of magnitude. Figure \ref{fig:rmse_tails} helps explain
these reductions. For EKF2, UKF, and CKF, median per-run RMSE is generally
similar across the plotted frameworks, while the 95th percentile often improves
substantially with probing. The benefit is therefore concentrated in
large-error trajectories. This pattern is consistent with the probe's
motivation: when local nonlinearity is weak, the moment map changes little
across the uncertainty region, and a single-center update can suffice.
When nonlinearity is strong or prediction uncertainty is severely
understated, probing can improve the correction and its covariance report.

For EKF, the improvement also extends to typical accuracy. EKF is the only
examined variant that discards all second- and higher-order measurement
curvature. In the noiseless quadratic case of
Proposition \ref{prop:understated_scale}, its residual
$\Omega_{\rm map}=0$ allows $\delta=0$, so the comparison already favors
probing at $\upsilon=1$ when $\Omega_q\ne0$.

Table \ref{tab:randomized} shows the same broad pattern across 300 setups. F4 has the lowest geometric-mean RMS of per-run RMSE for every KF variant in both applications, while its median is within 5\% of the best method in all eight combinations. Its 95th-percentile value is also the lowest in seven combinations and within 1\% of the lowest in the remaining battery CKF comparison.

\subsection{Covariance consistency}\label{sec:consistency}

Covariance consistency is another important metric for state estimators. In this work, covariance consistency is evaluated using the average normalized estimation error squared (ANEES) 
\cite{BarShalom2001,ANEES}. For key-state error $e_k^{(j)}$ and reported covariance block
$P_k^{(j)}$ in Monte Carlo run $j$, 
\begin{equation}\label{eq:anees}
\begin{aligned}
 \overline{\operatorname{NEES}}^{(j)}
 &=\frac{1}{(N_t+1)n_K}\sum_{k=0}^{N_t}
 e_k^{(j)\top}\big(P_k^{(j)}\big)^{-1}e_k^{(j)},\\
 \operatorname{ANEES}
 &=\frac{1}{N_{\rm MC}}\sum_{j=1}^{N_{\rm MC}}
 \overline{\operatorname{NEES}}^{(j)},
\end{aligned}
\end{equation}
where $N_{\rm MC}$ counts Monte Carlo runs, $N_t=60$ is the number of updates,
and $n_K=2$ is the number of key states: SOC and SOH for the battery, and both
position coordinates for navigation. ANEES averages all 61 stored time points,
including initialization, whereas RMSE uses the final 15 updates. ANEES near
one indicates agreement, on average, between reported covariance and observed error. 

\begin{figure*}[htbp]
\centering
\includegraphics[width=0.495\textwidth]{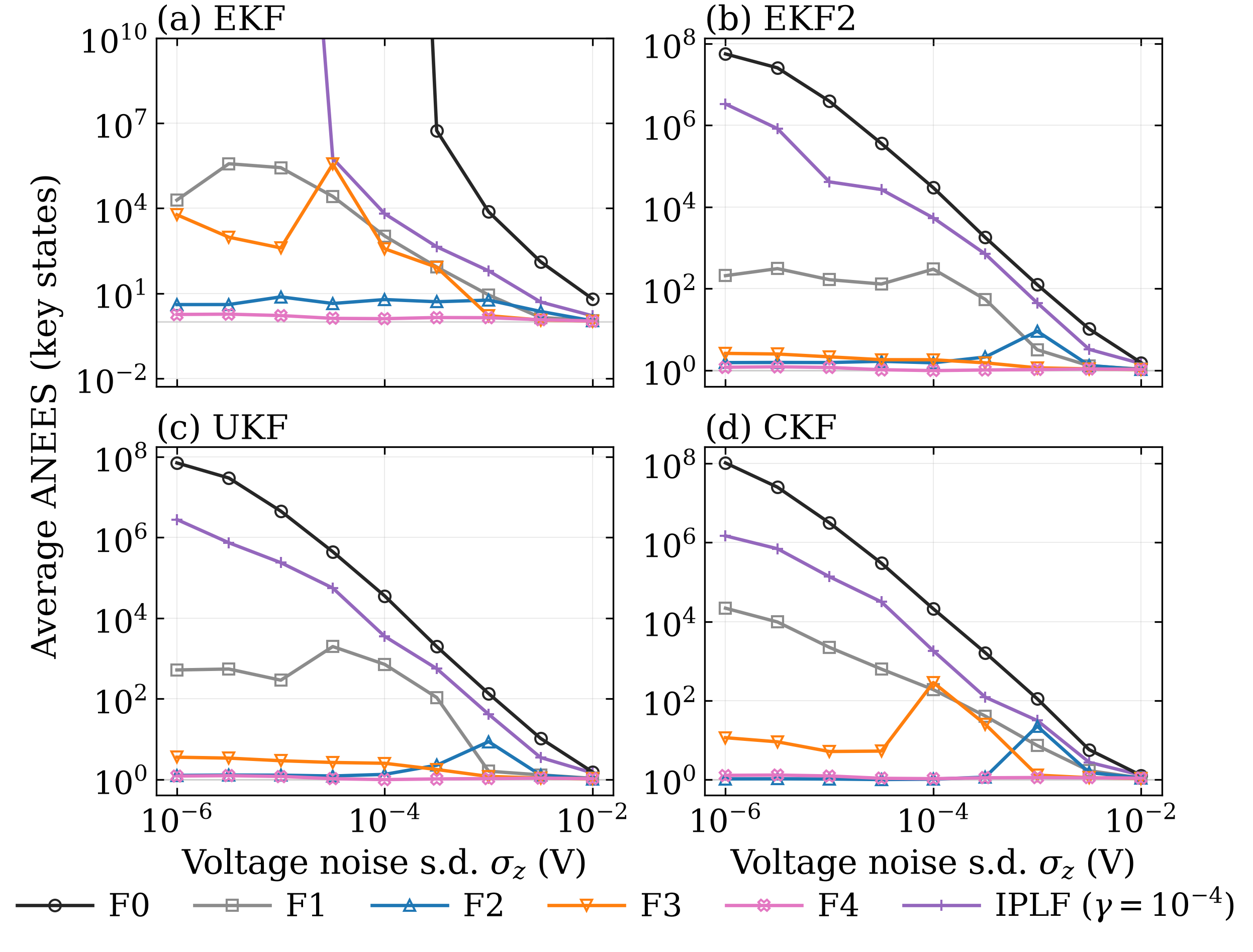}\hfill
\includegraphics[width=0.495\textwidth]{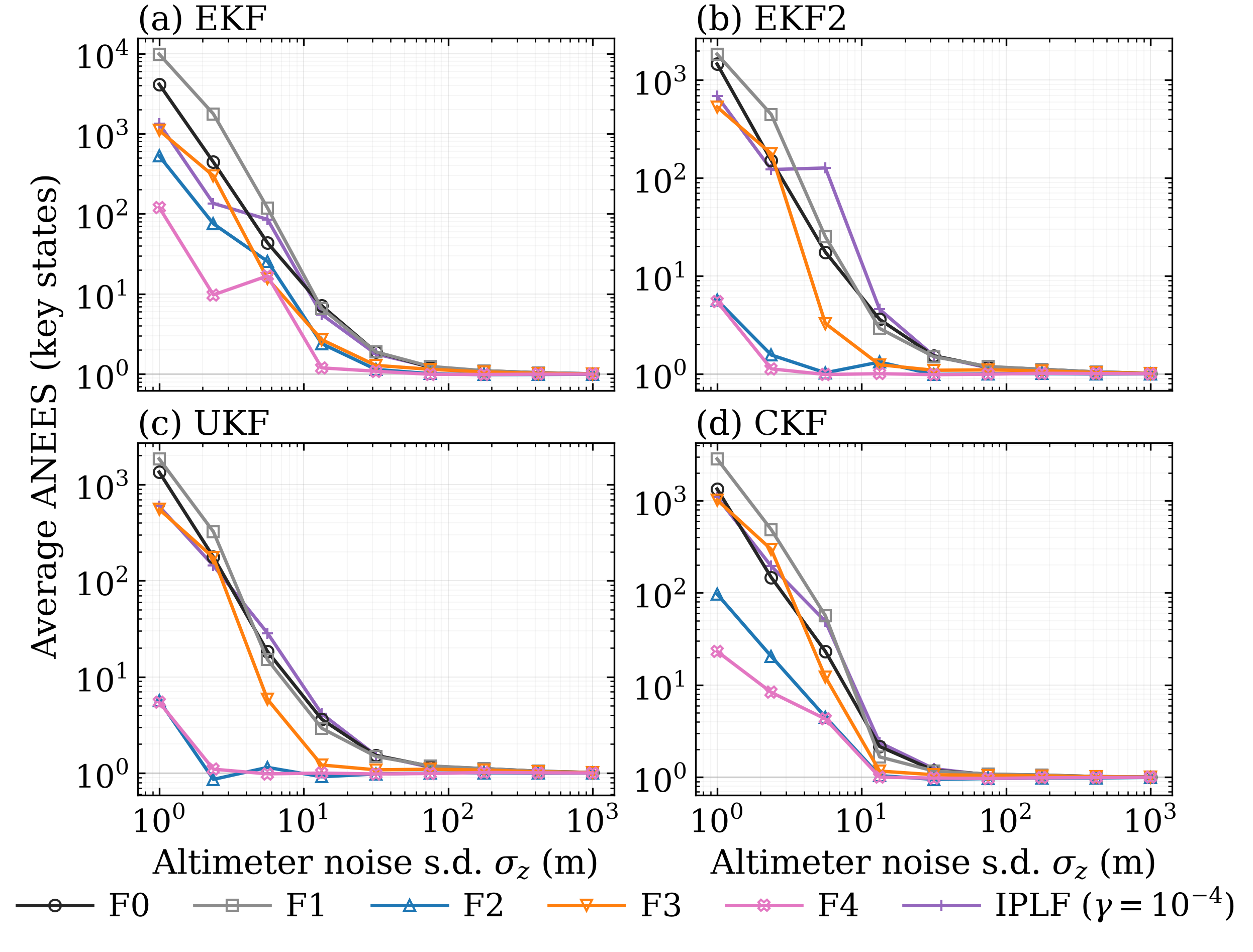}
\caption{Average key-state ANEES in the fixed measurement-noise sweeps.
Left: battery state estimation. Right: terrain-referenced navigation.
The battery EKF panel is limited to $10^{10}$ for readability. F0 and
IPLF exceed this limit at low measurement noise.}
\label{fig:anees_sweeps}
\end{figure*}

Figure \ref{fig:anees_sweeps} and Table \ref{tab:randomized} show that
probing substantially improves covariance consistency. In the battery
measurement-noise sweep, F4's ANEES ranges from $0.98$ to $1.85$. It is
the only principal method that stays within $[0.5,2]$ for every KF variant and noise level. The published baselines (F0, F1, and IPLF) can instead be overconfident by several orders of magnitude at low measurement noise.

For terrain-referenced navigation, a single exact elevation measurement
can correspond to several positions, making multimodal uncertainty
particularly relevant at low measurement noise. In
Fig. \ref{fig:anees_sweeps}, at $\sigma_z=1$ m,
F4 reduces ANEES to 120 for the EKF, about 5.5 for the EKF2 and UKF, and 23 for the CKF. F0 lies between $1.3\times10^3$ and $4.1\times10^3$, while F1 and IPLF are also strongly overconfident. Once $\sigma_z\ge13.3$ m, F4 lies between $0.97$ and $1.20$ for all four KF variants. The residual low-noise gap is consistent with the difficulty of approximating nonlinear measurement moments from only a mean and covariance when the predictive distribution is strongly non-Gaussian. F4's ANEES also remains above one for terrain-referenced navigation in Table \ref{tab:randomized}, yet the improvement over F0, F1, and IPLF is large.

\subsection{Ablation study}
Comparing F4 against F2 and F3 in Fig. \ref{fig:rmse_sweeps} and
Table \ref{tab:randomized} separates the roles of the two probes. F2 and F4 use the same probe-averaged covariance-report rule, so their difference isolates the update
probe. In Table \ref{tab:randomized}, F2 has the lowest ANEES in battery
estimation and the second-lowest in terrain-referenced navigation for every KF variant. However, its RMS and 95th-percentile RMSE are higher than F4's in all eight comparisons. Thus, a relatively consistent covariance report does not by itself ensure an accurate correction. Recalibration and back-out do not revise the selected Kalman gain. Repeated back-out can leave the filter relying mainly on prediction, whose error can remain large.

F3 and F4 use the same probe-mean gain rule, so their difference isolates the recalibration probe. F3 generally has low RMS, but its ANEES is higher than F4's in all eight comparisons in Table \ref{tab:randomized}. This result suggests that the recalibration probe primarily improves covariance consistency. Averaging the post-update moment maps retains the positive semidefinite slope-disagreement term in \eqref{eq:average_slope_form}, providing an additional contribution to the covariance report. Because the revised covariance feeds later predictions and gains, its benefit can extend to state accuracy.

\begin{table*}[!t]
\centering
\caption{Comparison with higher-order single-center baselines over 300 setups per application. Each per-run RMSE averages the two key-state RMSEs. Within each setup, we calculate the RMS, median, and 95th percentile of these values, and ANEES from \eqref{eq:anees}. Each entry is the geometric mean of the corresponding statistic over 300 setups. Battery RMSE is dimensionless. Navigation RMSE is in kilometers. ANEES has a nominal value of one. Within each two-row comparison, a value is bolded when it is more than 5\% below its partner's value. Rows with gray background indicate a KF variant with the proposed F4 framework.}
\label{tab:order_controls}
\scriptsize
\setlength{\tabcolsep}{2.2pt}
\renewcommand{\arraystretch}{1.16}
\setlength{\aboverulesep}{0pt}
\setlength{\belowrulesep}{0pt}
\vspace{3pt}
\begin{tabular}{lcccc|cccc}
\toprule
\multicolumn{5}{c|}{Battery state estimation}
& \multicolumn{4}{c}{Terrain-referenced navigation}\\
Method & RMS & Median & 95th percentile & ANEES
& RMS & Median & 95th percentile & ANEES\\
\midrule
EKF2 & $0.000841$ & $0.000293$ & $\boldsymbol{0.000938}$ & $2.00$ & $0.0799$ & $0.0261$ & $0.0673$ & $7.66$\\
\rowcolor{proposedgray}
EKF+probe & $\boldsymbol{0.000734}$ & $0.000291$ & $0.00103$ & $1.99$ & $0.0828$ & $0.0262$ & $0.0696$ & $\boldsymbol{7.04}$\\
\midrule
EKF3 & $0.000736$ & $0.000295$ & $0.000977$ & $1.77$ & $0.0730$ & $0.0262$ & $0.0602$ & $5.56$\\
\rowcolor{proposedgray}
EKF2+probe & $\boldsymbol{0.000574}$ & $0.000297$ & $\boldsymbol{0.000832}$ & $\boldsymbol{1.27}$ & $\boldsymbol{0.0657}$ & $0.0261$ & $0.0603$ & $\boldsymbol{4.23}$\\
\midrule
GH & $0.000778$ & $0.000290$ & $0.000895$ & $2.74$ & $0.0810$ & $0.0260$ & $0.0646$ & $7.14$\\
\rowcolor{proposedgray}
UKF+probe & $\boldsymbol{0.000596}$ & $0.000296$ & $\boldsymbol{0.000832}$ & $\boldsymbol{1.28}$ & $\boldsymbol{0.0634}$ & $0.0261$ & $\boldsymbol{0.0603}$ & $\boldsymbol{4.05}$\\
\midrule
CKF-5 & $0.000743$ & $\boldsymbol{0.000290}$ & $\boldsymbol{0.000888}$ & $2.29$ & $0.0808$ & $0.0261$ & $0.0636$ & $6.99$\\
\rowcolor{proposedgray}
CKF+probe & $\boldsymbol{0.000689}$ & $0.000311$ & $0.000993$ & $\boldsymbol{1.60}$ & $\boldsymbol{0.0711}$ & $0.0262$ & $0.0632$ & $\boldsymbol{5.38}$\\
\bottomrule
\end{tabular}
\end{table*}

\subsection{Comparison to higher-order single-center filters}\label{sec:order_controls}

Probing uses additional computation to evaluate the moment map at several
centers. Higher-order methods instead refine the approximation at one
center. To compare these two uses of computation, we pair EKF2 with
EKF+probe and the third-order EKF (EKF3) with EKF2+probe. Gauss--Hermite
(GH) quadrature and fifth-degree cubature (CKF-5) are paired with
UKF+probe and CKF+probe, respectively. Here ``+probe'' denotes the
proposed framework (F4).
Each higher-order baseline retains the predictor of its paired lower-order filter, changes
only the measurement moment rule, and uses single-point recalibration with
the unit-invariant back-out rule ($M=P^{-1}$). The comparison therefore contrasts higher local
order at one center with multi-center evaluation of an existing moment rule.
The EKF pair exemplifies these two uses of computation: the update probe
uses $n+1$ local linearizations, whereas differences between such Jacobians
can instead approximate the Hessians needed by EKF2 when analytic Hessians
are unavailable. GH scales differently, since its node count grows exponentially
with the state dimension.
Table \ref{tab:order_controls} reports the same four statistical summaries
as Table \ref{tab:randomized}.

F4 yields a lower RMS of per-run RMSE in seven of the eight comparisons,
with reductions of approximately $7$--$23\%$. The exception is EKF+probe in terrain-referenced navigation, whose RMS is $3.6\%$ higher than that of EKF2. 

Median RMSE differs by less than $2.1\%$ in seven of the eight comparisons. The exception is CKF+probe in battery estimation, whose median is $7.3\%$ higher than that of CKF-5. The 95th-percentile results are mixed: F4 improves this statistic by more than $5\%$ in three comparisons, while the higher-order baseline improves it by more than $5\%$ in two.
Thus, probing improves RMS more consistently than the 95th percentile.
RMS gives greater weight to the largest per-run errors, which can differ
substantially even when the medians and 95th percentiles are similar.

The consistency comparison further highlights the probe's strength. F4 has a geometric-mean ANEES closer to its nominal value of one in all eight comparisons, with reductions above $5\%$ in seven of eight cases. The largest reductions occur for UKF+probe relative to GH:
from $2.74$ to $1.28$ in battery estimation and from $7.14$ to $4.05$ in
terrain-referenced navigation. The comparison of EKF3 with EKF2+probe also shows that probing can improve RMS and ANEES when the underlying moment map already includes second-order terms.

\subsection{Computational complexity and runtime}\label{sec:runtime}

Finally, we compare moment-map evaluation counts and measured runtime. A center-point
update requires one measurement
moment evaluation. F4 performs
$n+1$ evaluations for the update probe and $n+1$ for the recalibration
probe, together with one center evaluation for the predicted measurement. It
therefore uses $2n+3$ moment-map calls per measurement update in the implementation. F2 and F3 use
$n+2$ and $n+3$ calls, respectively, while F1 uses two center calls.

\begin{figure}[!t]
\centering
\includegraphics[width=\columnwidth]{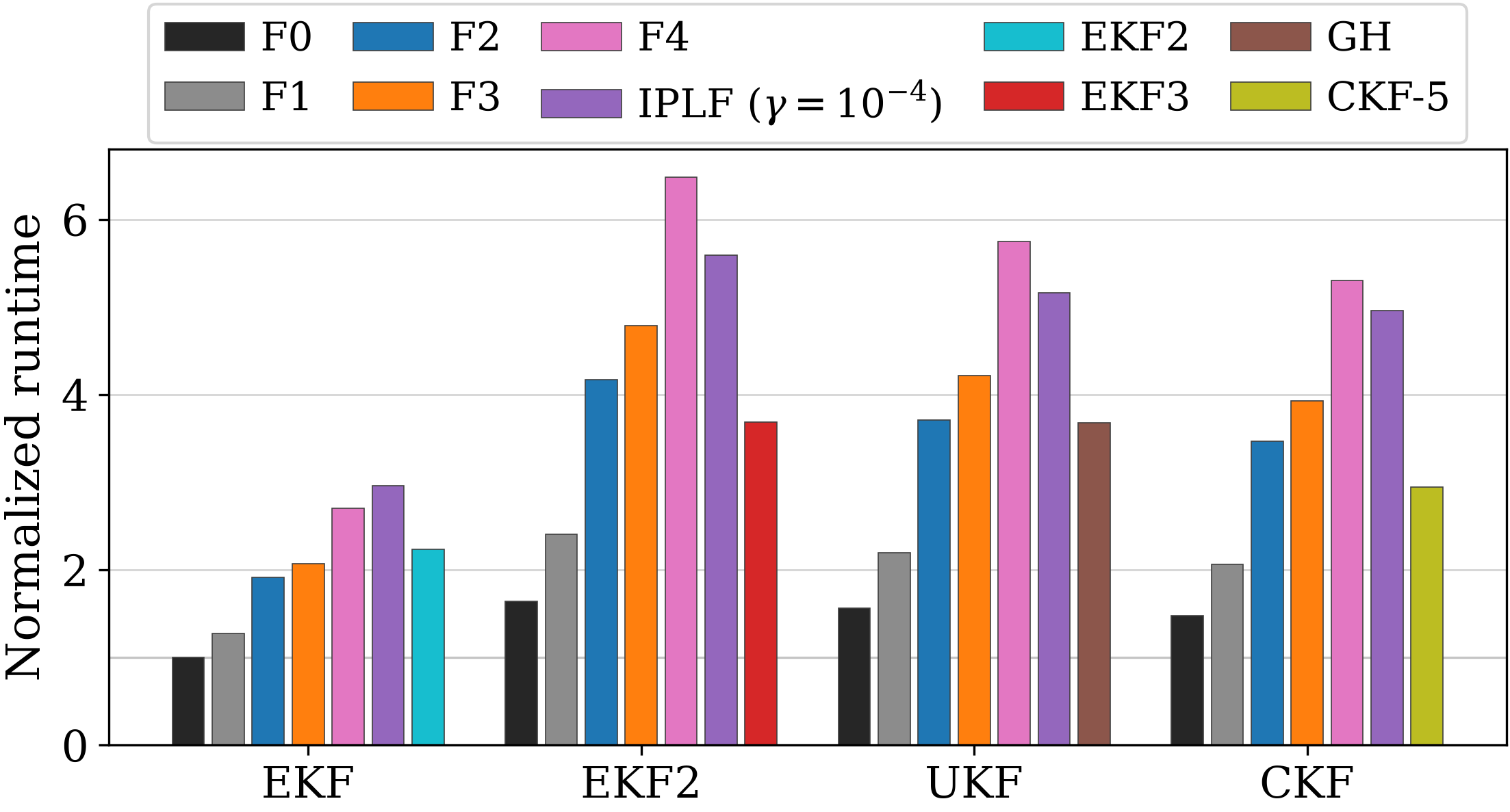}
\caption{Runtime normalized by the conventional EKF within each application
and then averaged over the two applications. Timings use $\sigma_z=10^{-4}$ V
for battery estimation and $\sigma_z=1$ m for terrain-referenced navigation. Each bar
uses 200 serial trajectories. Each filter group includes its corresponding
higher-order single-center baseline: EKF2, EKF3, GH, or CKF-5.}
\label{fig:runtime}
\end{figure}

Figure \ref{fig:runtime} shows that F4 costs approximately $2.7$, $6.5$,
$5.8$, and $5.3$ times the common conventional-EKF baseline when wrapping
EKF, EKF2, UKF, and CKF, respectively. Relative to the conventional
implementation of the same filter, its runtime is approximately
$2.7$--$4.0$ times as large. F4 takes approximately $0.9$--$1.2$ times
the runtime of IPLF and requires no online iteration. The higher-order
baselines cost approximately $2.2$--$3.7$ times the common baseline.
Table \ref{tab:order_controls} reports the accuracy and covariance consistency
obtained at these costs.

\section{Conclusion}\label{sec:conclusion}

Conventional nonlinear KFs approximate measurement moments around a single
center, which can lead to inaccurate corrections and overconfident covariance estimates when the measurement model varies substantially across the uncertainty region. Covariance-scaled probe sets account for this variation in gain selection and covariance recalibration. We show that the probe-mean gain minimizes the average reported
error covariance across moment approximations evaluated at the
probe points. For quadratic measurements under an exact Gaussian prediction and the stated assumptions, the probe update does not increase
the true estimation-error covariance, and its update-probe covariance report upper-bounds the actual posterior covariance. We also establish conditions under which probing reduces the total normalized true error variance relative to the single-center update when prediction uncertainty is understated. Across 300 randomized setups per application in battery estimation and terrain-referenced navigation, the framework achieves the lowest RMS of per-run RMSE among the six principal methods while retaining near-best median and 95th-percentile values. Compared with the paired higher-order single-center baselines, it yields ANEES closer to one in all eight comparisons and lower RMS of per-run RMSE in seven. Its runtime is approximately $2.7$--$4.0$ times that of the
corresponding conventional filter. Future work will explore the use of
probe sets in other state estimation methods, including moving horizon
estimation.

\bibliographystyle{unsrt}
\bibliography{autosam}

\appendix

\section{Implemented nonlinear Kalman filters}\label{app:filters}

Let $g:\mathbb R^n\rightarrow\mathbb R^q$ be evaluated about $\xi$ with
spread $P\succeq0$. For this general transform, $q$ denotes the output
dimension and $G=\nabla g(\xi)$ its Jacobian. When $g=h_k$, these become
$q=p$ and the measurement Jacobian used in the main text.
Each nonlinear-KF variant returns
\begin{equation}\label{eq:transform}
 \mathcal T_g(\xi,P)=(\hat g,C_{xg},P_g),
\end{equation}
where $\hat g$ approximates $\mathbb E[g(x)]$, $C_{xg}$ approximates
$\operatorname{Cov}(x,g(x))$, and $P_g$ approximates
$\operatorname{Var}(g(x))$. For prediction, $g(x)=f_k(x,u_{k-1})$, with the known input held fixed, and
\begin{equation}\label{eq:prediction_transform}
 \hat x_{k|k-1}=\hat g,\qquad P_{k|k-1}=P_g+Q_{k-1}.
\end{equation}
For measurement evaluation, $g=h_k$ and
\begin{equation}\label{eq:measurement_transform}
 \hat z(\xi)=\hat g,\qquad C_P(\xi)=C_{xg},\qquad S_P(\xi)=P_g+R_k.
\end{equation}

For the EKF, with $G:=\nabla g(\xi)$,
\begin{equation}\label{eq:ekf_transform}
 \hat g=g(\xi),\qquad C_{xg}=PG^\top,
 \qquad P_g=GPG^\top.
\end{equation}

For the EKF2, let $G_a^{(2)}:=\nabla^2g_a(\xi)$, where $a,b=1,\ldots,q$
index output components.
With $G_{a:}$ denoting row $a$ of $G$, the implemented moments are
\begin{align}
 [\hat g]_a&=g_a(\xi)+\tfrac12\operatorname{tr}(G_a^{(2)}P),
 \label{eq:ekf2_mean}\\
 C_{xg}&=PG^\top,\label{eq:ekf2_cross}\\
 [P_g]_{ab}&=G_{a:}PG_{b:}^\top
 +\tfrac12\operatorname{tr}(G_a^{(2)}PG_b^{(2)}P).
 \label{eq:ekf2_cov}
\end{align}

The EKF3 moment rule computes the Gaussian moments of the third-order Taylor
polynomial. Let \(T_a:=\nabla^3g_a(\xi)\) and define
\(\tau_a\in\mathbb R^n\) by
\begin{equation}\label{eq:ekf3_q}
 [\tau_a]_i:=\sum_{j,k=1}^{n}(T_a)_{ijk}P_{jk}.
\end{equation}
Let \(\mathcal T_P\in\mathbb R^{q\times n}\) have \(\tau_a^\top\) as its \(a\)th row.
The mean remains \eqref{eq:ekf2_mean}, while
\begin{equation}\label{eq:ekf3_cross}
 C_{xg}=P\left(G+\tfrac12\mathcal T_P\right)^\top.
\end{equation}
The transformed covariance is
\begin{align}
 [P_g]_{ab}={}&G_{a:}PG_{b:}^\top
 +\tfrac12\operatorname{tr}(G_a^{(2)}PG_b^{(2)}P)\notag\\*
 &+\tfrac12\left(G_{a:}P\tau_b+G_{b:}P\tau_a\right)
 +\tfrac14\tau_a^\top P\tau_b\notag\\*
 &+\tfrac16\sum_{\substack{i,j,k\\l,m,r}}
 (T_a)_{ijk}P_{il}P_{jm}P_{kr}(T_b)_{lmr}.
 \label{eq:ekf3_cov}
\end{align}
These expressions include every Gaussian moment generated by the cubic Taylor
polynomial and are exact when \(g\) is cubic.
The summation indices $i,j,k,l,m,r$ in \eqref{eq:ekf3_cov} range from $1$ to $n$.

The UKF uses the scaled unscented transform of \cite{UKF2} with
\begin{equation}\label{eq:ukf_parameters}
 \alpha=10^{-3},\qquad \beta=2,\qquad \kappa=0.
\end{equation}
Its $2n+1$ sigma points and associated mean and covariance weights are applied to the joint
state-output transform in \eqref{eq:transform}. The CKF uses the standard $2n$
equally weighted spherical-radial rule of \cite{CKF}. Both filters compute
$\hat g$, $C_{xg}$, and $P_g$ as the weighted sample mean, cross-covariance,
and output covariance of their transformed points.

The EKF2 and EKF3 baselines retain EKF and EKF2 prediction, respectively.
Each evaluates its measurement rule before and after the correction.

For the quadrature rules below, each dimensionless point $u\in\mathbb R^n$
is converted to a state point $\xi+Lu$, where $LL^\top=P$.
Multiplication by $L$ sets the spread, and adding $\xi$ places the points
around the evaluation center.
The GH quadrature rule uses all $m^n$ Cartesian products of the probabilists' Hermite
nodes and normalized weights, with $m=3$ for battery and $m=4$ for terrain. It
is exact for Gaussian polynomial moments through degree $2m-1$ in each
coordinate \cite{QKF}. The CKF-5 cubature rule uses the $2n^2+1$ fifth-degree rule of
\cite{CKF5}: the origin has weight $2/(n+2)$, axis points
$\pm\sqrt{n+2}\,e_i$ have weight $(4-n)/[2(n+2)^2]$, and diagonal points
$\sqrt{(n+2)/2}(\pm e_i\pm e_j)$, for $1\le i<j\le n$ and all sign
choices, have weight $1/(n+2)^2$. Here $e_i$ is the $i$th coordinate unit vector.
Both baselines retain
the lower-order predictor, replace only the measurement moments, and call the
measurement rule twice per update.

\section{Detailed simulation setup}\label{app:systems}
This appendix specifies the fixed measurement-noise sweeps and the
randomized setups.
Both systems are adapted from our previous work \cite{parent}.

In each Monte Carlo run, the initial estimate is generated as
$x_0+\epsilon_0$, where $x_0$ is the true initial state and
$\epsilon_0\sim\mathcal N(0,P_0)$. The battery initial estimates are
then clipped as described below, with the reported $P_0$ retained.
Process and measurement noises are Gaussian, mutually independent,
and independent across time and of the initial draw. The filters use
the same $P_0$, $Q$, and measurement-noise covariance $R$ as the simulation.
To preserve positive semidefiniteness in numerical computations, we
propagate covariance factors $L$ satisfying $P=LL^\top$.

\subsection{Battery state estimation}

The battery model has state $x=[\mathrm{SOC},U_c,\mathrm{SOH}]^\top$ and a
terminal-voltage measurement. The state $U_c$ is the polarization-branch
voltage, measured in volts. SOC is remaining capacity relative to present
maximum capacity. SOH is present maximum capacity relative to initial capacity.
Both SOC and SOH are dimensionless.

The fixed setup uses the current profile, in amperes,
\begin{equation}\label{eq:battery_current}
I_k =
\begin{cases}
0, & 0\le k\le15,\\
-1.3\dfrac{k-15}{30}, & 15<k\le45,\\
-1.3, & 45<k\le60.
\end{cases}
\end{equation}
The cell begins at 35\% SOC and ends near 10\% SOC under the nominal model.

The parameters are $Q_0=1$ Ah, $R_1=0.01\,\Omega$, $R_2=0.05\,\Omega$,
$1/(R_2C_1)=0.008\,\mathrm{s}^{-1}$, and $\Delta t=20$ s. Here $Q_0$ is the nominal cell capacity, $R_1$ and
$R_2$ are resistances, and $C_1$ is the polarization capacitance.
The open-circuit voltage (OCV) is
\begin{equation}\label{eq:ocv}
\begin{aligned}
\mathrm{OCV}=\sum_{i=0}^9\Big(&\tfrac{\mathrm{SOH}-0.8}{0.2}\,a_{100,10-i}\\
&+\tfrac{1-\mathrm{SOH}}{0.2}\,a_{80,10-i}\Big)\mathrm{SOC}^i,
\end{aligned}
\end{equation}
where
\begin{equation*}
\begin{aligned}
a_{100}={}&[1390.38,-6961.31,14760.31,-17230.92,\\
&12055.71,-5162.75,1330.60,-196.37,\\
&15.60,2.96],\\
a_{80}={}&[813.94,-4229.96,9345.49,-11415.38,\\
&8396.15,-3801.07,1043.09,-165.29,\\
&14.28,2.96].
\end{aligned}
\end{equation*}
The state and measurement models are
\begin{equation}\label{eq:battery_state}
\begin{bmatrix}\mathrm{SOC}_{k}\\U_{c,k}\\\mathrm{SOH}_{k}\end{bmatrix}
=
\begin{bmatrix}
\mathrm{SOC}_{k-1}+\dfrac{I_{k-1}\Delta t}{3600Q_0\mathrm{SOH}_{k-1}}\\
\eta U_{c,k-1}+(1-\eta)R_2I_{k-1}\\
\mathrm{SOH}_{k-1}
\end{bmatrix}+w_{k-1},
\end{equation}
\begin{equation}\label{eq:battery_measurement}
 z_k=\mathrm{OCV}(\mathrm{SOC}_k,\mathrm{SOH}_k)+U_{c,k}+R_1I_k+v_k,
\end{equation}
where $\eta=e^{-\Delta t/(R_2C_1)}$. The independent noises are
\begin{equation}\label{eq:battery_noise}
 w_{k-1}\sim\mathcal N(0,Q),\qquad
 v_k\sim\mathcal N(0,\sigma_z^2).
\end{equation}
The prior scale $\rho_P$ and process scale $\rho_Q$ multiply the
base standard deviations of the initial estimation error and process
noise, respectively. The resulting covariance matrices are
\begin{align}
 P_0&=\rho_P^2\operatorname{diag}(0.05^2,0.005^2,0.03^2),
 \label{eq:battery_P0}\\
 Q&=(2\times10^{-6}\rho_Q)^2\operatorname{diag}(1,1,0).
 \label{eq:battery_Q}
\end{align}
Here the entries follow the state order and units defined above.
SOH is constant and carries no process noise.
The initial SOC and SOH estimates are clipped to $[0.02,0.98]$
after the Gaussian draw in both studies.

\emph{Fixed setup.}
The true initial state is $(0.35,0,0.86)$, with
$\rho_P=1.7$ and $\rho_Q=0.3$. Only the measurement-noise
standard deviation $\sigma_z$ varies.

\emph{Randomized setups.}
For each setup, $\sigma_z$, $\rho_P$, and $\rho_Q$ are sampled
independently and log-uniformly from $[10^{-6},10^{-2}]$ V,
$[0.25,4]$, and $[0.25,4]$, respectively. Log-uniform means
uniform on a logarithmic scale.
The true initial SOH, $\mathrm{SOH}_0$, and current amplitude,
$I_{\max}$, are sampled independently from uniform distributions
on $[0.84,0.96]$ and $[1,3]$ A. Replacing $1.3$ by $I_{\max}$
in \eqref{eq:battery_current} gives the known current profile.
The rest and ramp durations are unchanged.

Initial SOC is sampled after $\mathrm{SOH}_0$ and $I_{\max}$.
Let $d_{\mathrm{SOC}}=I_{\max}(\tfrac16\,\mathrm{h})/(Q_0\mathrm{SOH}_0)$
be an upper bound on the nominal SOC decrease over the trajectory.
The initial SOC is uniform on
$[\max\{0.25,d_{\mathrm{SOC}}+0.05\},0.75]$, which keeps the
nominal trajectory above 5\% SOC. The true initial polarization
voltage is $U_{c,0}=0$. All other model parameters retain their fixed-setup values.
The sampled parameters and true initial state are held fixed across
the Monte Carlo runs within each setup.

\subsection{Terrain-referenced navigation}

The state is the two-dimensional aircraft position in kilometers, and the
measurement is terrain elevation in meters. The fixed sweep uses
\begin{equation}\label{eq:terrain}
 h(x)=\sum_{j=0}^{3}\mathcal A_j\sin(k_j^\top x+\phi_j),\qquad
 \mathcal A_j=200\,2^{-0.97j}\ \mathrm{m}.
\end{equation}
Here $\mathcal A_j$ is the amplitude, $k_j$ is the wave vector, and $\phi_j$
is the phase of terrain component $j$. The wavelengths are $(20,10,5,2.5)$ km,
the directions are
$(238.4,42.1,288.0,112.5)^\circ$, and the phases are
$(289.4,313.3,92.0,197.3)^\circ$. If $\lambda_j$ and $\psi_j$ denote the
$j$th wavelength and direction, then
\begin{equation*}
 k_j=\frac{2\pi}{\lambda_j}
 \begin{bmatrix}\cos\psi_j&\sin\psi_j\end{bmatrix}^\top.
\end{equation*}

The state and measurement models are
\begin{equation}\label{eq:terrain_state}
 \begin{bmatrix}x_{1,k}\\x_{2,k}\end{bmatrix}
 =\begin{bmatrix}x_{1,k-1}+0.25\\x_{2,k-1}+0.15\end{bmatrix}+w_{k-1},
\end{equation}
\begin{equation}\label{eq:terrain_measurement}
 z_k=h(x_k)+v_k.
\end{equation}
The simulation has 60 steps, with true initial state $(5,5)$ km.
The prior scale $\rho_P$ multiplies the base initial standard deviation
of $0.30$ km in each position coordinate. The process scale $\rho_Q$
multiplies the base process-noise standard deviation of $0.002$ km
per step in each coordinate. Thus,
\begin{align}
 P_0&=(0.30\rho_P)^2I_2\ \mathrm{km}^2,\label{eq:terrain_P0}\\
 Q&=(0.002\rho_Q)^2I_2\ \mathrm{km}^2,\qquad
 R=\sigma_z^2.\label{eq:terrain_QR}
\end{align}
The measurement-noise standard deviation $\sigma_z$ is in meters.

\emph{Fixed setup.}
The scales are $\rho_P=1.5$ and $\rho_Q=0.25$, giving
$P_0=0.2025I_2\ \mathrm{km}^2$ and
$Q=2.5\times10^{-7}I_2\ \mathrm{km}^2$.
Only $\sigma_z$ varies in the measurement-noise sweep.

\emph{Randomized setups.}
For each setup, $\sigma_z$, $\rho_P$, and $\rho_Q$ are sampled
independently and log-uniformly from $[0.1,10^3]$ m, $[0.5,2]$,
and $[0.25,4]$, respectively.
The terrain amplitudes are $\mathcal A_j=\mathcal A_0 2^{-Hj}$,
where the roughness (Hurst) exponent $H$ and largest-component
amplitude $\mathcal A_0$ are sampled independently and uniformly
from $[0.75,0.99]$ and $[120,320]$ m.
The four directions and four phases are independently uniform
on $[0,2\pi)$. These map parameters and the noise and prior scales
are held fixed across the Monte Carlo runs within each setup.
The wavelengths, true initial position, and nominal motion retain
their fixed-setup values.

\end{document}